\documentclass[12pt,a4paper]{article}
\usepackage[english]{babel}
\usepackage{amsmath,amssymb,amsthm,epsfig,color,graphicx, mathtools, bbm, dsfont, geometry, float, mathrsfs, bm, hyperref}
\usepackage[latin1]{inputenc}
\usepackage[T1]{fontenc}
\usepackage{indentfirst}
\usepackage[shortlabels]{enumitem}
\usepackage[title]{appendix}
\usepackage{lmodern}       
\usepackage{comment}  
\usepackage{textcomp, upgreek}

\usepackage{tikz, tkz-euclide}
\usetikzlibrary{patterns}
\usetikzlibrary{arrows.meta}
\usetikzlibrary{calc}
\usetikzlibrary{decorations.markings}
\usepackage{fancybox}

\usepackage{graphicx}
\graphicspath{ {./images/} }
\newcommand{\Z}{\mathbb{Z}}

\newcommand{\lab}{\mathrm{lab}}     
\newcommand{\din}{\partial_{\mathrm{in}}}

\newcommand{\diam}{\mathrm{diam}}
\newcommand{\dis}{\mathrm{dist}}

\newcommand{\I}{\mathrm{I}}
\newcommand{\Sp}{\mathrm{sp}}
\newcommand{\ssum}[1]{
\sum_{\mathclap{\substack{#1}}}
}

\newcommand{\be}{\begin{equation}}
\newcommand{\ee}{\end{equation}}

\numberwithin{equation}{section}

\usepackage{bbm}
  \newcounter{dummy} \numberwithin{dummy}{section}
  \theoremstyle{plain}
  \newtheorem*{theorem*}        {Theorem}
	\newtheorem*{conjecture*}   {Conjecture}
  \newtheorem{theorem}[dummy]          {Theorem}
  \newtheorem{lemma}[dummy]              {Lemma}
  \newtheorem*{lemma*}          {Lemma}

  \newtheorem{proposition}[dummy]       {Proposition}
   
  \newtheorem{remark}[dummy]           {Remark}
  \theoremstyle{remark}
    
  \theoremstyle{definition}
   \newtheorem{definition}[dummy]          {Definition}
\makeatletter

\newcommand\longleftrightarrowfill@{%
  \arrowfill@\leftarrow\relbar\rightarrow}
\makeatother

\definecolor{Red}{cmyk}{0,1,1,0}

\definecolor{Blue}{cmyk}{1,1,0,0}

\definecolor{DarkBlue}{rgb}{0.1,0.1,0.5}
\definecolor{Red}{rgb}{0.9,0.0,0.1}
\definecolor{DarkGreen}{rgb}{0.10,0.50,0.10}
\definecolor{DarkRed}{rgb}{0.50,0.10,0.10}
\definecolor{bleu}{RGB}{0,140,189}%

\newgeometry{vmargin={15mm}, hmargin={12mm,17mm}}

\begin{document}

\begin{center}
{\LARGE Long-Range Ising Models: Contours, Phase Transitions and Decaying Fields}
\vskip.5cm
Lucas Affonso$^{1}$, Rodrigo Bissacot$^{1}$, Eric O. Endo$^{2}$, Satoshi Handa$^{3}$
\vskip.3cm
\begin{footnotesize}
$^{1}$Institute of Mathematics and Statistics (IME-USP), University of S\~{a}o Paulo, Brazil\\
$^{2}$NYU-ECNU Institute of Mathematical Sciences, NYU Shanghai, China\\
$^{3}$Department of Mathematics, Hokkaido University and Fujitsu Laboratories Ltd., Japan.
\end{footnotesize}
\vskip.1cm
\begin{scriptsize}
emails: lucas.affonso.pereira@gmail.com, rodrigo.bissacot@gmail.com, ericossamiendo@gmail.com, handa.math@gmail.com
\end{scriptsize}

\end{center}

\begin{abstract}
Inspired by Fr\"{o}hlich-Spencer and subsequent authors who introduced the notion of contour for long-range systems, we provide a definition of contour  and a direct proof for the phase transition for ferromagnetic long-range Ising models on $\Z^d$, $d\geq 2$. The argument,  which is based on a multi-scale analysis, works for the sharp region  $\alpha>d$ and improves previous results obtained by Park for $\alpha>3d+1$, and by Ginibre, Grossmann, and Ruelle for $\alpha> d+1$, where $\alpha$ is the power of the coupling constant.  The key idea is to avoid a large number of small contours. As an application, we prove the persistence of the phase transition when we add a polynomially decaying magnetic field with power $\delta>0$ as $h^*|x|^{-\delta}$, where $h^* >0$.  For $d<\alpha<d+1$,  the phase transition occurs when $\delta>\alpha-d$,  and when $h^*$ is small enough over the critical line $\delta=\alpha-d$.  For $\alpha \geq d+1$,   $\delta>1$ is enough to prove the phase transition, and for $\delta=1$ we have to ask $h^*$ small.  The natural conjecture is that this region is also sharp for the phase transition problem when we have a decaying field.
\end{abstract}

	\section{Introduction}

Most of our knowledge about statistical mechanics of lattice systems comes from short-range interactions and one of the main problems in the area is to decide if a model presents or not a phase transition. The standard strategy to prove the phase transition is to define a notion of contour for the model to apply the \textit{Peierls' argument} \cite{Pei}.  However, many important examples have long-range interactions.  In this paper, we consider the ferromagnetic long-range Ising model with an external field in dimension $d\geq 2$. The Hamiltonian of the model is given formally by
\be\label{hamiltonian}
H = -\sum_{x,y \in \Z^d}J_{xy}\sigma_x\sigma_y - \sum_{x \in \Z^d}h_x \sigma_x ,
\ee 
where $J_{xy} = \widetilde{J}(|x-y|)= J|x-y|^{-\alpha}$, $J>0$,  $\alpha > d$ and $h_x \in \mathbb{R}$. The case considered originally was the zero magnetic field, that is, $h_x=0$ for every $x \in \Z^d$. 

Part of the results for ferromagnetic long-range Ising models in dimension one ($d=1$) can be summarized as follows:  Kac and Thompson conjectured in \cite{KT} that the model exhibits a phase transition at low temperatures when  $\alpha \in (1,2]$. The conjecture was proved in 1969 by Dyson in \cite{Dyson} for $\alpha \in (1,2)$.  Dyson's approach applies correlation inequalities between a model that is known nowadays as \emph{hierarchical model}. In a private communication with Dyson, Thompson told that there was no phase transition for $\alpha=2$ (See the references in \cite{Dyson}). Still, in 1982, Fr\"{o}hlich and Spencer \cite{Fro1} introduced a notion of one-dimensional contour and proved the phase transition based on a Peierls type argument.  

The contours presented in \cite{Fro1} were inspired by techniques introduced by the same authors in \cite{Fro3} to study two-dimensional systems with continuous symmetries, where breakthroughs were made in the study of the Berezin-Kosterlitz-Thouless transition. The idea consists in organizing the spin flips in (not necessarily connected) contours that satisfy a condition related to their distance between each other (in \cite{Fro2}, it is called \emph{Condition D}). 

Cassandro, Ferrari, Merola and Presutti \cite{Cass} attempted to extend the contour argument to different exponents for the interactions introducing a more geometric approach in terms of triangle configurations to the problem of the phase transition. Their modification comes with an additional condition on the coupling for nearest neighbours, namely, $\widetilde{J}(1) \gg 1$. The authors showed the phase transition assuming that $\alpha \in (2-\alpha^+,2]$, where ${\alpha^+ =\log(3)/\log(2) - 1}$. Many results were obtained following either the contour approach of Fr\"{o}hlich and Spencer or Cassandro, Ferrari, Merola and Presutti, such as cluster expansion \cite{Cass2, Imbr1, Imbr2}, phase separation \cite{Cass3}, phase transition for the one-dimensional long-range model with a random field \cite{Cass1} and condensation for the corresponding lattice gas model \cite{Jo}. 

Recently, Bissacot, Endo, van Enter, Kimura, and Ruszel \cite{Eric2}, based on the contour argument of \cite{Cass}, considered the model with a presence of the decaying magnetic field $(h_x)_{x\in \Z}$. Moreover, restricting the region of $\alpha$ to $(2-\alpha^*,2]$ where $\alpha^+<\alpha^*$ satisfies $\sum_{n \geq 1}n^{-\alpha^*}=2$, they removed the condition of the nearest neighbour's coupling $\widetilde{J}(1)\gg 1$.

For multidimensional models ($d\geq 2$), one of the only notions of contour available beyond the short-range case were proposed by Park \cite{Park1, Park2}, on an extension of Pirogov-Sinai theory.  Park's arguments work for pair interactions and finite state space, but the interaction's exponent has to be assumed $\alpha>3d+1$. Before this, in \cite{GGR}, Ginibre,  Grossmann, and Ruelle extended the usual Peierls argument for a class of interactions more general than nearest-neighbor. Their result implies that, if we add a long-range perturbation to the nearest-neighbor model of the form $K_{x,y} = K|x-y|^{-\alpha}$, with $K$ a real constant, the cost of erasing a contour $\gamma$ from a family of contours $\Gamma$ containing $\gamma$ is bounded below by  $2\left(J -  |K| c_\alpha\right)|\gamma|$, where $c_\alpha = \sum_{k\in \Z^d\setminus\{0\}}|k|^{1-\alpha} - 2d$. Notice that this lower bound is positive for $K$ small enough and $\alpha>d+1$. Indeed, their result holds in more generality: the constant $K$ may be replaced by a family of real constants $K_{x,y}$ with $\sup_{\{x,y\}\in \Z^d}|K_{x,y}|$ small enough. Notice that the sign of $K_{x,y}$ may change, making their result applicable for antiferromagnetic long-range perturbations, for instance. Since the methods of Ginibre-Grossmann-Ruelle hold for antiferromagnetic perturbations, the results by van Enter \cite{vanEnter81} and Biskup, Chayes, and Kivelson \cite{Bisk} show that their enhanced Peierls argument is not generalizable for $d<\alpha\leq d+1$, since an arbitrarily small antiferromagnetic long-range interaction with this decay forces the magnetization to be zero.

One of our main contributions in this paper is to present a proof for the phase transition at low temperature for the \emph{ferromagnetic long-range Ising model on the lattice $\Z^d$} with $d\ge 2$.  Our proof combines ideas from \cite{ Cass, Fro3, Fro2, Park1, Park2, Pres},  and we are able to obtain the Peierls argument in the case of zero magnetic field for the sharp region $\alpha > d$. Notice that the positivity of the magnetization follows from the GKS inequalities and the Peierls argument for the nearest neighbor Ising model. The novelty in our results resides in the fact that we have a contour argument applicable to models where the correlation inequalities do not imply the phase transition or even do not hold, such as the model with a decaying field in the former case and to random magnetic fields \cite{Johannes}, in the later.

It is well known that ferromagnetic Ising models do not present phase transition in the presence of a uniform magnetic field; see Lee and Yang \cite{LY1,LY2}.  When the field is not constant, the situation drastically changes.  There is some literature about models with fields,  including the famous case of the random field. See, for instance, \cite{BBL,BBCK,BEE,Bov, BK,CMR,GPY,Preston}.

The problem of studying the phase diagram with a decaying field was introduced in \cite{Bis1}.  Since the pressure of the Ising model with a decaying magnetic field $(h_x)_{x\in \Z^d}$ is equal to the pressure of the Ising model with zero magnetic field, it may induce the belief that these two models should present the same phase diagram.  In \cite{Bis2}, the authors studied the phase diagram of the ferromagnetic nearest-neighbor Ising model on the lattice $\Z^d$ in the presence of the spatially dependent magnetic field $(h_x)_{x\in \Z^d}$ given by $h_x=h^*|x|^{-\delta}$, when $x\neq 0$, and $h_0=h^*$, where $h^*$ and $\delta$ are positive constants.  They studied the behaviour of the model at low temperatures according to the exponent of the decaying field: the model undergoes a phase transition at low temperature for $\delta>1$, and we have uniqueness for  $0<\delta<1$. Afterward, Cioletti and Vila \cite{CV} closed the gap, concluding the uniqueness of the Gibbs measure when $0<\delta<1$ for all temperatures, the argument uses the Fortuin-Kasteleyn representation.

When we add a decaying magnetic field,  it is not possible to obtain the phase transition for the long-range Ising model from the short-range case, so the natural strategy is to use the Peierls argument, as we see in \cite{Bis1,Bis2,Eric2}. In fact, these models show that arguments using contours are sharp with respect to the decaying of the magnetic field. 

In the present paper, we extended the analysis in \cite{Bis2} from the short-range to the long-range case considering decaying external fields in the Hamiltonian \eqref{hamiltonian}. By a similar approach of Fr\"{o}hlich and Spencer, we define a notion of contour on this model to show the phase transition at low temperature when $d<\alpha<d+1$ and $\delta>\alpha-d$, and when $\alpha\ge d+1$  and $\delta>1$. 

To understand the heuristics of our result on phase transitions, let us consider the configurations
\be
\sigma_x = \begin{cases}
	+1& \text{if } x \in B_R(z), \\
	-1& \text{otherwise},
\end{cases}
\ee
where $B_R(z)\subset \Z^d$ is the closed ball in the $\ell_1$-norm centered in $z$ with radius $R\geq 0$. Let $\Omega_c$ be the collection of all such configurations together with the configuration that is $-1$ in all $\Z^d$ and, for fixed $\Lambda \Subset \Z^d$, let $\Omega_{c,\Lambda}$ be the subset of configurations where $B_R(z)\subset \Lambda$ together with the configuration that is $-1$ in all $\Z^d$. Then, we have
\be\label{contaheuristica}
\mu_{\beta,\Lambda}^-( \sigma_0 = +1|\Omega_{c,\Lambda}) \leq \sum_{R\geq 1}R^d \exp\left(-\beta (cR^{d-1}+ F_{B_R}- \sum_{x \in B_R(0)}h_x)\right),
\ee
where the quantity $F_{B_R}$ is defined as
\be
F_{B_R} \coloneqq \sum_{\substack{x \in B_R(0)\\ y \in B_R(0)^c}}J_{xy}.
\ee
One can understand this quantity as a surface energy term, and it has different asymptotics depending on the parameters $\alpha$ and $d$.  Denoting by $f\approx g$  the fact that given two functions $f,g:\mathbb{R}^{+} \rightarrow \mathbb{R}^{+}$, there exist positive constants $A'\coloneqq A'(\alpha, d),  A\coloneqq A(\alpha, d)$ such that $A'f(x) \leq g(x) \leq Af(x)$ for every $x>0$ large enough, we have
\be
F_{B_R} \approx \begin{cases}
	R^{2d-\alpha}&\text{if } d<\alpha<d+1, \\
	R^{d-1}\log(R)&\text{if } \alpha = d+1, \\
	R^{d-1}&\text{if } \alpha>d+1.
\end{cases}
\ee
For the proof, see Propositions 3.1 and 3.4 in  \cite{BBCK}.  For our purposes, we will need estimates for more general subsets than balls, see Lemma \ref{lema3}.   Now, the inequality \eqref{contaheuristica} show us that the phase transition occurs when $\delta> \alpha - d$.  To see this, observe that the surface energy term should be  larger than the contribution of field, which is given by
\be
\sum_{x \in B_R(0)}h_x = O(R^{d-\delta}).
\ee

Our main result can be summarized by the following picture:

\begin{figure}[H]
	\centering
	
	\tikzset{every picture/.style={line width=0.75pt}} 
	
	\begin{tikzpicture}[scale=1.25]
		
		
		\draw[-{Stealth}, black] (-1,0)--(5,0);
		\draw[-{Stealth}, black] (0,-1)--(0,4);
		\draw[-{Stealth}, black] (6,2.6)to[out=180, in=0](5,2);
		
		\draw[-, dashed] (0,0)--(2.5,2)--(5,2);
		\fill[color= black, opacity= 0.2] (0,0)--(2.5,2)--(5,2)--(5,4)--(0,4)-- cycle;	
		
		\fill[white, rounded corners, thick] (1.5,3) rectangle (3.5,3.4);
		\draw[fill=white, draw=black, rounded corners, thick] (1.5,3) rectangle (3.5,3.4);
		\draw[black, rounded corners, thick] (3,1) rectangle (4.5,1.5);
		\draw[black, rounded corners, thick] (6,2.3) rectangle (8.1,2.9);
		
		
		\draw (0,0) node[anchor=north east, scale=1] {$0$};
		\draw (5,0) node[anchor=north east, scale=1] {$\alpha -d$};
		\draw (0,4) node[anchor=north east, scale=1] {$\delta$};
		\draw (0.1,2)--(-0.1,2);
		\draw (0,2) node[anchor=north east, scale=1] {$1$};
		\draw (2.5,0.1)--(2.5,-0.1);
		\draw (2.5,0) node[anchor=north west, scale=1] {$1$};
		\draw(1.45,3.05) node[anchor=south west, scale=0.7] {\;\small\textbf{Phase Transition}};
		\draw(2.9,1.05) node[anchor=south west, scale=0.7] {\;\textbf{Uniqueness?}};
		\draw(5.9,2.6) node[anchor=south west, scale=0.7] {\;\textbf{Phase Transition}};
		\draw(6.3,2.3) node[anchor=south west, scale=0.7] {\;\textbf{for small $h^*$}};

	\end{tikzpicture}
	\label{figpt}
	\caption{The phase diagram for the long-range Ising model depending on $\alpha$ and $\delta$.}
\end{figure}
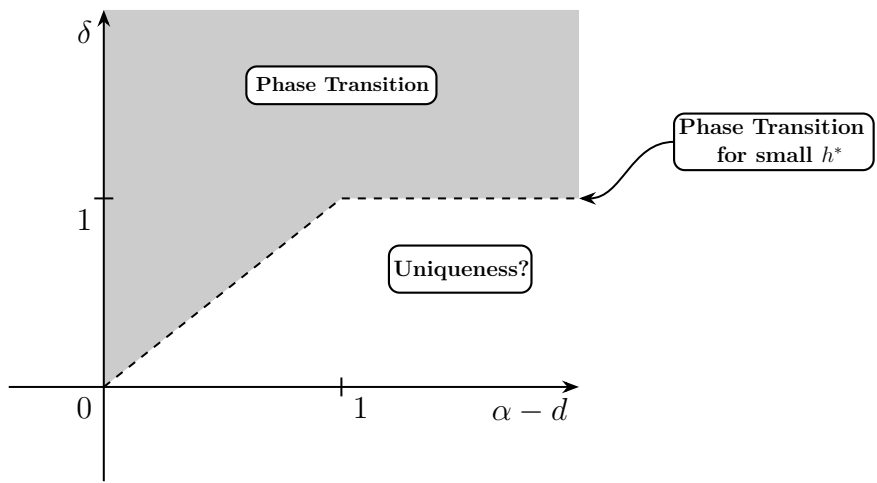

This paper is divided as follows. In Section 2, we give some basic definitions and define the model.  In section 3,  we introduce the notion of contour for the long-range Ising model and estimates for the entropy of the contours.  In the Section 4, we present the proof of the phase transition for the region indicated in the Figure 1.

This paper is contained in the Ph.D. thesis of Lucas Affonso \cite{Affonso}.

\section{Definitions and notations}

\subsection{The model}
Although most of the definitions and results cited in this section are available in greater generality, we choose to consider only the case where the lattice is $\Z^d$ for $d\geq 2$ and the state space is $E=\{-1,+1\}$. Our configuration space will be $\Omega = \{-1,+1\}^{\Z^d}$. We will denote $\sigma^+$ (resp. $\sigma^-$) the configuration which assume the same value equal to $+1$ (resp. $-1$) in all vertices, i.e., $\sigma^{\pm}_x=\pm 1$ for all $x\in\Z^d$. We write $\Lambda \Subset \Z^d$ for finite subsets of the lattice. Fixing such $\Lambda$, we define the local configuration space $\Omega_\Lambda = \{-1,+1\}^{\Lambda}$ and, given $\omega \in \Omega$, we define $\Omega_\Lambda^\omega \subset \Omega$ as the subset of configurations such that $\sigma_{\Lambda^c}=\omega_{\Lambda^c}$, i.e., $\sigma_x=\omega_x$ for all $x\in \Lambda^c$.  We say that $\sigma$ has $+$-boundary condition (resp. $-$-boundary condition) when $\sigma\in \Omega_{\Lambda}^{+}$ (resp. $\sigma\in \Omega_{\Lambda}^{-}$).  

For each $\Lambda \Subset \Z^d$, $\omega \in \Omega$, and $\mathbf{h}=(h_x)_{x\in \mathbb{Z}^d}$ a collection of real numbers, the \emph{Hamiltonian function of the long-range Ising model} $H^\omega_{\Lambda,\mathbf{h}}:\Omega^{\omega}_\Lambda \rightarrow \mathbb{R}$ is given by
\begin{equation}\label{Isingsys}
	H_{\Lambda,\mathbf{h}}^\omega(\sigma) = -\sum_{\{x,y\} \subset \Lambda} J_{xy} \sigma_x\sigma_y - \sum_{\substack{x \in \Lambda \\ y \in \Lambda^c}}J_{xy} \sigma_x \omega_y - \sum_{x \in \Lambda} h_x \sigma_x,
\end{equation}
where, for $\alpha>d$, $J>0$ and $x,y\in \Z^d$ with $x\neq y$, the \emph{coupling constant}  $J_{xy}$ is given by
\begin{equation}\label{long}
	J_{xy} = \begin{cases}
		\frac{J}{|x-y|^\alpha} &\text{ if } x\neq y, \\
		0 & \text{otherwise},
	\end{cases}
\end{equation}
where $|z|$ is the $\ell^1$ norm of $z\in \Z^d$. For $\delta,h^*>0$, the \emph{magnetic field} $(h_x)_{x \in \Z^d}$ is given by
\begin{equation}\label{magfield}
	h_x = \begin{cases}
		h^* & \text{ if }x=0, \\
		\frac{h^*}{|x|^\delta} &\text{ if }  x\neq 0.
	\end{cases}
\end{equation}

For a subset $\Lambda\subset \Z^d$, consider the $\sigma$-algebra $\mathcal{F}_{\Lambda}$ generated by the cylinder sets supported on $\Lambda$. The basic object of the study in the classical statistical mechanics is the  collection of probability measures in $(\Omega,\mathcal{F}_{\Z^d})$ called \emph{finite volume Gibbs measures} defined by
\begin{equation}
	\mu_{\beta,\mathbf{h}, \Lambda}^\omega (\sigma) \coloneqq 
	\frac{e^{-\beta H_{\Lambda,\mathbf{h}}^\omega(\sigma)}}{Z_{\beta,\mathbf{h}, \Lambda}^\omega} \text{   if } \sigma_{\Lambda^c}=\omega_{\Lambda^c},
\end{equation}
and $\mu_{\beta,\mathbf{h}, \Lambda}^\omega (\sigma) =0$ otherwise. Here $\beta>0$ is the inverse temperature, and the normalization factor $Z^\omega_{\beta,\mathbf{h}, \Lambda}$, called \emph{partition function}, is defined by
\be
Z^\omega_{\beta,\mathbf{h}, \Lambda} \coloneqq \sum_{\sigma \in \Omega_\Lambda}e^{-\beta H^\omega_{\Lambda, \mathbf{h}}(\sigma)}.
\ee

The Gibbs measure with $+$ (resp. $-$) boundary condition is when $\omega_x=+1$ (resp. $-1$) for every $x \in \Z^d$. We say that a sequence of finite subsets $(\Lambda_n)_{n \in \mathbb{N}}$ \emph{invades} $\Z^d$, denoted by $\Lambda_n \nearrow \Z^d$, if for every finite subset $\Lambda$ there exists $N>0$ such that $\Lambda \subset \Lambda_n$ for all  $n \geq N$. Since the configuration space $\Omega$ is compact, the space of probability measures is a weak* compact set. Therefore the net defined by the collection of finite volume Gibbs measures has a convergent subsequence. We define the set $\mathcal{G}_\beta$ as the closed convex hull of all the limits obtained by this procedure, i.e.,
\be
\mathcal{G}_\beta \coloneqq \overline{\text{conv}}\{\mu_\beta: \mu_\beta = w^*\text{-}\lim_{\Lambda' \nearrow \Z^d}\mu_{\beta,\Lambda'}^\omega, \Lambda' \Subset \Z^d \text{ invades the lattice}\}.
\ee

The set $\mathcal{G}_{\beta}$ is always non empty in our case by a simple application of the Banach-Alaoglu theorem. We say that the model has \emph{uniqueness at $\beta$} if $|\mathcal{G}_\beta|= 1$ and it undergoes to a \emph{phase transition at $\beta$} if~$|\mathcal{G}_\beta|> 1$.

\section{Contours}

Contours are geometric objects first introduced in a seminal paper of R. Peierls \cite{Pei}. The technique is known nowadays as \emph{Peierls' Argument}. Many attempts were made to extend the ideas of the Peierls argument to the other systems \cite{Dat,  DobCont,Park1, Pirogov, Zaradnik}. The most sucessful generalization was made by S. Pirogov and Y. Sinai in \cite{Pirogov}, and later improved by  Zahradn\'{i}k \cite{Zaradnik}. Their work is known as \emph{Pirogov-Sinai} theory, and it is applied to models with short range interactions that may not have symmetries.


Park, in \cite{Park1, Park2} extended the theory of Pirogov-Sinai to systems with two-body long-range interactions that satisfy a condition equivalent to \eqref{long} having decay $\alpha > 3d+1$. Inspired by \cite{Fro1}, in this section we will introduce new contours more suitable to study long-range two-body interactions. 

\begin{definition}\label{def1}
	Given $\sigma \in \Omega$, a point $x \in \Z^d$ is called \emph{+ (or - resp.)} \emph{correct} if $\sigma_y = +1$, (or $-1$, resp.) for all points $y$ in $B_1(x)$. The \emph{boundary} of $\sigma$, denoted by $\partial \sigma$, is the set of all points in $\Z^d$ that are neither $+$ nor $-$ correct.
\end{definition}

We recall the reader that $B_1(x)$ is the ball of radius $1$ with respect to the $\ell^1$-norm. Note that the boundary can be an infinite subset of $\Z^d$. Indeed, if we take $\sigma \in \Omega$ defined by
\[
\sigma_x = \begin{cases}
	+1& \text{if } |x|\text{ is even}, \\
	-1& \text{otherwise}.
\end{cases}
\]
It is easy to see that every point in $\Z^d$ is incorrect with respect to the configuration $\sigma$, and thus $\partial\sigma = \Z^d$. This situation can be avoided by restricting our attention only to configurations with finite boundary set. By definition of incorrectness, for a configuration $\sigma$ to have a finite boundary set it must satisfy $\sigma \in \Omega_\Lambda^+$ or $\Omega_\Lambda^-$ for some subset $\Lambda \Subset \Z^d$. 

For each subset $\Lambda \Subset \Z^d$, we call $\Lambda^{(0)}$ the unique unbounded connected component of $\Lambda^c$. Then, we define the \emph{volume} by $V(\Lambda)= (\Lambda^{(0)})^c$. Note that this set is a union of simply connected sets that contains $\Lambda$ and is the smallest one in the subset order. The \emph{interior} is defined by $\I(\Lambda)= \Lambda^c \setminus \Lambda^{(0)}$.

In Pirogov-Sinai theory, the construction of the contours starts by considering first the connected subsets of the boundary $\partial\sigma$. This procedure is troublesome for long-range models since each point of the lattice interacts with all the other points. Thus the contours always have a nonvanishing interaction between themselves. To avoid this problem, we will divide the boundary of a configuration in a way where the interaction between them will be negligible in a sense to be specified later. Inspired by Fr\"{o}hlich and Spencer \cite{Fro1}, we introduce the following definition. 
\begin{definition}\label{def:d_condition}
	Fix real numbers $M,a,r>0$. For each configuration $\sigma \in \Omega$ with finite boundary $\partial\sigma$, a set $\Gamma(\sigma) \coloneqq \{\overline{\gamma} : \overline{\gamma} \subset \partial\sigma\}$ is called an $(M,a,r)$-\emph{partition} when the following conditions are satisfied:
	\begin{enumerate}[label=\textbf{(\Alph*)}, series=l_after] 
		\item They form a partition of $\partial\sigma$, i.e.,  $\displaystyle\bigcup_{\overline{\gamma} \in \Gamma(\sigma)}\overline{\gamma}=\partial\sigma$ and $\overline{\gamma} \cap \overline{\gamma}' = \emptyset$ for distinct elements of $\Gamma(\sigma)$. 
		\item For all $\overline{\gamma} \in \Gamma(\sigma)$ there exist $1\leq n \leq 2^r-1$ and a family of subsets $(\overline{\gamma}_{k})_{1\leq k \leq n}$ satisfying 
		\begin{enumerate}[\textbf{(B\arabic*)}]
			\item  $\displaystyle \overline{\gamma} = \bigcup_{1\leq k \leq n}\overline{\gamma}_{k}$,
			\item For all distinct $\overline{\gamma},\overline{\gamma}' \in \Gamma(\sigma)$,
			\be\label{B_distance}
			\dis(\overline{\gamma},\overline{\gamma}') > M \min\left \{\underset{1\leq k \leq n}{\max}\diam(\overline{\gamma}_k),\underset{1\leq j\leq n'}{\max}\diam(\overline{\gamma}'_{j})\right\}^a,
			\ee
			where $(\overline{\gamma}'_j)_{1\leq j \leq n'}$ is the family given by item $\textbf{(B1)}$ for $\overline{\gamma}'$.
		\end{enumerate} 
	\end{enumerate}
\end{definition}
Note that the sets $\overline{\gamma} \in \Gamma(\sigma)$ may be disconnected. We also stress that in Condition \textbf{(B1)} the sets $(\overline{\gamma}_k)_{1\leq k \leq n}$ may not be disjoint. Some of the next results are true for any $M,a,r>0$, as the existence of $(M,a,r)$-partition for any configuration $\sigma$ with finite boundary $\partial\sigma$,  see Proposition \ref{prop1}. 
\begin{remark}
    For the main purposes of this paper, which is the proof of the phase transition, the constant $a$ is chosen as $a \coloneqq a(\alpha, d) = \max\left\{\frac{d+1+\varepsilon}{\alpha-d}, d+1+\varepsilon\right\}$, for some $\varepsilon>0$ fixed and $r$ given by $r= \lceil \log_2(a+1) \rceil + d +1$, where $\lceil x \rceil$ is the smallest integer greater than or equal to $x$. The motivation for these choices will be clear in the proofs. The specific value of $M\coloneqq M(\alpha, d)$ will be chosen later, but we will assume for now on that $M\geq 1$.
\end{remark}
\begin{figure}[H]
	\centering	
	\tikzset{every picture/.style={line width=0.75pt}} 
	\begin{tikzpicture}[scale=1.25]
		\begin{scope}[shift={(-5,0)}]
			\draw [ line width=0.5mm, draw=black, dashed] plot [smooth cycle] coordinates {(0.25,1.75) (1.1,2) (1.25,2.75) (0.25,3) }--cycle;
			\draw [line width=0.5mm, fill=black!9!white, draw=black] plot [smooth cycle] coordinates {(0.5,2.25) (1,2.5) (0.5,2.75) (0.25,2.5)}--cycle;
			\draw [line width=0.5mm, fill=white, draw=black, dashed] plot [smooth cycle] coordinates {(1,-0.5) (5.5,-0.5) (5.5,1.5) (4,1.5) (3.5,0.75) (1,0.75) }--cycle;
			\draw [line width=0.5mm, fill=black!8!white, draw=black] plot [smooth cycle] coordinates {(1.5,0) (2,0) (2.5,0.5) (1.5,0.5) (1.25,0.25) }--cycle;
			\draw [line width=0.5mm, fill=black!8!white, draw=black] plot [smooth cycle] coordinates {(3.25,0.5) (3.25,-0.25) (4,-0.25) (4,0.5) }--cycle;
			\draw [line width=0.5mm, fill=black!8!white, draw=black] plot [smooth cycle] coordinates {(3.25,0.5) (3.25,-0.25) (4,-0.25) (4,0.5) }--cycle;
			\draw [line width=0.5mm, fill=white, draw=black] plot [smooth cycle] coordinates {(3.75,0.3) (3.75,0) (3.5,0)  (3.5,0.3) }--cycle;
			\draw [line width=0.5mm, fill=black!8!white, draw=black] plot [smooth cycle] coordinates {(4.75,1.5) (4.75,1) (5.5,0.5) (5.5,1) }--cycle;
			\draw (0.9,2.1) node[scale=1] {$\overline{\gamma}$};
			\draw (4.8,0.2) node[scale=1] {$\overline{\gamma}'$};
			\draw[<->] (0.5,2.25)--(1.5,0.5);
			\draw (2.2,1.4) node[scale=1] {$\dis(\overline{\gamma},\overline{\gamma}') = \frac{n^a}{2}$};
			\draw[<->] (2.2,0.05)--(3.15,0);
			\draw (2.675,-0.15) node[scale=1] {$n^{a-1}$};
			\draw[<->] (4.05,0.5)--(4.9,0.8);
			\draw (4.45,0.85) node[scale=1] { $n^{a-1}$};
		\end{scope}
		\draw [-, dashed] (2,-1)--(2,4);
		\begin{scope}[shift={(4,0)}]
			\draw [ line width=0.5mm, draw=black, dashed] plot [smooth cycle] coordinates {(-1.25,1.75) (1.1,1.8) (1.25,2.75) (-0.5,3) }--cycle;
			\draw [line width=0.5mm, fill=black!8!white, draw=black] plot [smooth cycle] coordinates {(-1,2) (1,2.5) (0.5,2.75) (-0.8,2.5)}--cycle;
			\draw [line width=0.5mm, fill=white, draw=black, dashed] plot [smooth cycle] coordinates {(1,-0.5) (4,-0.5) (4.5,-1) (5.7,-1) (5.5,1.5) (4,1.5) (3.5,0.75) (1,0.75) }--cycle;
			\draw [line width=0.5mm, fill=black!8!white, draw=black] plot [smooth cycle] coordinates {(1.5,0) (2,0) (2.5,0.5) (1.5,0.5) (1.25,0.25) }--cycle;
			\draw [line width=0.5mm, fill=black!8!white, draw=black] plot [smooth cycle] coordinates {(3.25,0.5) (3.25,-0.25) (4,-0.25) (4,0.5) }--cycle;
			\draw [line width=0.5mm, fill=black!8!white, draw=black] plot [smooth cycle] coordinates {(3.25,0.5) (3.25,-0.25) (4,-0.25) (4,0.5) }--cycle;
			\draw [line width=0.5mm, fill=white, draw=black] plot [smooth cycle] coordinates {(3.75,0.3) (3.75,0) (3.5,0)  (3.5,0.3) }--cycle;
			\draw [line width=0.5mm, fill=black!8!white, draw=black] plot [smooth cycle] coordinates {(4.75,1.5) (4.75,1) (5.5,0.5) (5.5,1) }--cycle;
			\draw [line width=0.5mm, fill=black!8!white, draw=black] plot [smooth cycle] coordinates {(5,-0.3) (4.75,-0.9) (5.5,-0.9) (5.5,-0.3) }--cycle;
			\draw (0.9,2.1) node[scale=1] {$\overline{\gamma}$};
			\draw (4.8,0.2) node[scale=1] {$\overline{\gamma}'$};
			\draw[<->] (0.5,2.25)--(1.5,0.5);
			\draw (2.2,1.4) node[scale=1] {$\dis(\overline{\gamma},\overline{\gamma}') = 2n^a$};
			\draw[<->] (2.2,0.05)--(3.15,0);
			\draw (2.675,-0.15) node[scale=1] { $n^a$};
			\draw[<->] (5.3,-0.2)--(5.2,0.6);
			\draw (5.5,0.2) node[scale=1] { $n^2$};
                \draw[<->] (4.05,0.5)--(4.9,0.8);
			\draw (4.45,0.85) node[scale=1] { $n^2$};
			\draw[<->] (4.1,0)--(4.85,-0.4);
			\draw (4.5,-0.5) node[scale=1] { $n^2$};
		\end{scope}	
	\end{tikzpicture}
	\caption{Consider $M=1, r=2$, and $a>2$. For the image on the left, consider that all the connected components (grey regions) of $\overline{\gamma}'$ have diameters equal to $n$ and $\diam(\overline{\gamma})=3n$. In this case, there is no family of subsets of $\overline{\gamma}'$ satisfying condition \textbf{(B)}. A possible $(M,a,r)$-partition for this case is $\Gamma(\sigma)=\{\overline{\gamma}\cup \overline{\gamma}'\}$. For the figure on the right, consider that all the connected components of $\overline{\gamma}'$ have diameter $n$ and $\diam(\overline{\gamma})=3n^2$. Notice that, in this case, the family of subsets of $\overline{\gamma}'$ satisfying Inequality \eqref{B_distance} must have $n' >2^r-1$.}
\end{figure}
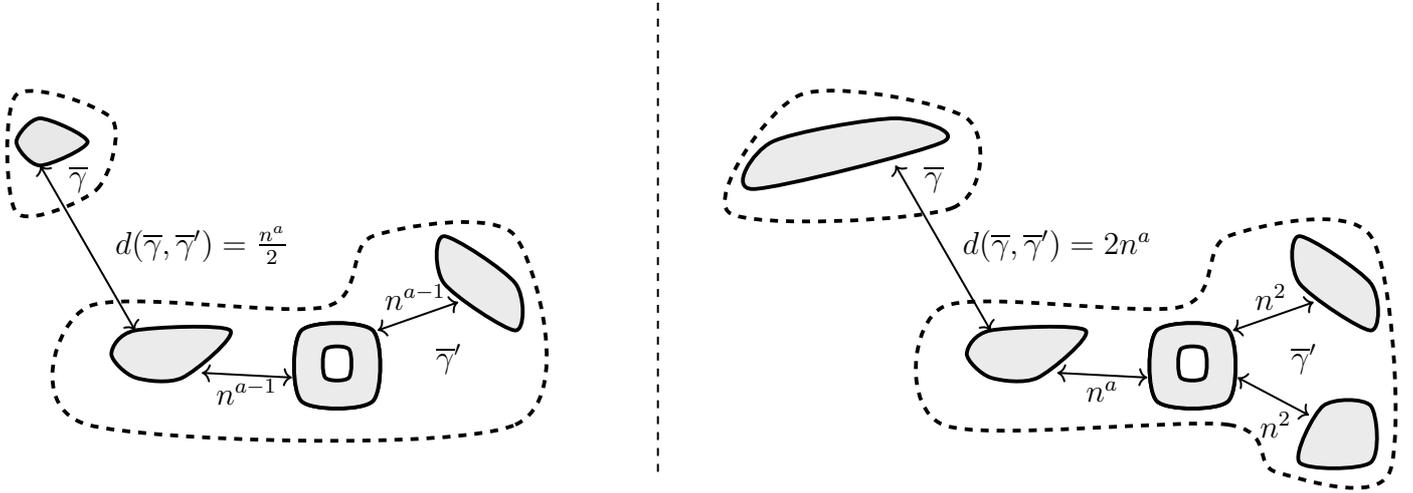
In general, there are many possible $(M,a,r)$-partitions for a given $\sigma \in \Omega$ with finite boundary and fixed $M,a,r>0$. If we have two $(M,a,r)$-partitions $\Gamma$ and $\Gamma'$, we say that $\Gamma$ \emph{is finer than} $\Gamma'$ if for every $\overline{\gamma} \in \Gamma$ there is $\overline{\gamma}' \in \Gamma'$ with $\overline{\gamma} \subseteq \overline{\gamma}'$. The finest $(M,a,r)$-partition also satisfies stronger separation properties than what is stated in Condition \textbf{(A)}. It will actually satisfy the following condition (see also Figure \ref{A1}),
\begin{itemize}
    \item[\textbf{(A1)}] For any $\overline{\gamma},\overline{\gamma}^\prime\in \Gamma(\sigma)$, $\overline{\gamma}'$ is contained in only one connected component of $(\overline{\gamma})^c$.
\end{itemize}
The proposition below is important in the proof that there is a finest $(M,a,r)$-partition.
\begin{proposition}
Let $\sigma$ be a configuration with finite $\partial \sigma$, and let $\Gamma,\Gamma'$ be two $(M,a,r)$-partitions. Then the set defined as
\[
\Gamma\cap \Gamma'\coloneqq \{\gamma \cap \gamma': \gamma \in \Gamma, \gamma' \in \Gamma', \gamma \cap \gamma' \neq \emptyset\},
\]
is also an $(M,a,r)$-partition. Moreover, if one of them satisfies condition \textbf{(A1)}, then $\Gamma\cap \Gamma'$ also satisfies condition \textbf{(A1)}.
\end{proposition}
\begin{proof}
	The existence of non-trivial $(M,a,r)$-partitions will be given in the next section. Consider two $(M,a,r)$-partitions $\Gamma$ and $\Gamma'$. We will show that we can build an $(M,a,r)$-partition that is finer than $\Gamma$ and $\Gamma'$. We will prove that $\Gamma \cap \Gamma'$ is an $(M,a,r)$-partition. It is easy to see that $\Gamma\cap \Gamma'$ forms a partition of $\partial \sigma$. For condition \textbf{(B)}, consider $\overline{\gamma}_i\cap \overline{\gamma}'_i\in \Gamma\cap\Gamma'$ and choose as a family of subsets $(\overline{\gamma}_i\cap \overline{\gamma}'_i)_k \coloneqq \overline{\gamma}_i\cap \overline{\gamma}'_{i,k}$ for all $k \in\{1,\dots,n'\}$ where the intersection is not empty. Here $(\overline{\gamma}'_{i,k})_{1\leq k\leq n'}$ is the family of subsets for $\overline{\gamma}'$ given by condition \textbf{(B)}. We have,
	\begin{align*}
	\dis(\overline{\gamma}_1\cap\overline{\gamma}'_1,\overline{\gamma}_2\cap\overline{\gamma}'_2) \geq \dis(\overline{\gamma}'_1, \overline{\gamma}'_2) &> M \min\left \{\underset{1\leq k \leq n'}{\max}\diam(\overline{\gamma}'_{1,k}),\underset{1\leq j\leq m'}{\max}\diam(\overline{\gamma}'_{2,j})\right\}^a \\
	&\geq M \min\left \{\underset{1\leq k \leq n'}{\max}\diam((\overline{\gamma}_1\cap\overline{\gamma}_1')_k),\underset{1\leq j\leq m'}{\max}\diam((\overline{\gamma}_2\cap\overline{\gamma}_2')_j)\right\}^a.
	\end{align*}
Let $\overline{\gamma}_i \in \Gamma$ and $\overline{\gamma}'_i \in \Gamma'$ such that $\overline{\gamma}_i\cap \overline{\gamma}_i' \in \Gamma \cap \Gamma'$ for $i=1,2$. Suppose that $\Gamma$ satisfies the Condition \textbf{(A1)}. Then, there is a connected component $A$ of $(\overline{\gamma}_2)^c$ that contains $\overline{\gamma}_1$. Hence,
	\[
	\overline{\gamma}_1 \cap \overline{\gamma}'_1\subset A \subset (\overline{\gamma}_2 \cap \overline{\gamma}'_2)^c.
	\]   
\end{proof}
\begin{remark}
Since there are a finite number of $(M,a,r)$-partitions for a finite boundary set $\partial \sigma$, we can construct the finest one by intersecting all of them, following the above construction. If we had two $(M,a,r)$-partitions that have the property of being the finest, we could produce another $(M,a,r)$-partition finer than both, yielding a contradiction.    
\end{remark}
\begin{remark}
	In their one-dimensional paper \cite{Fro1}, Fr\"{o}hlich and Spencer assumed that they would choose the finest partition of spin flips satisfying their Condition D, but the uniqueness of this partition would not be a problem, since you just need to fix a partition for the phase transition argument to hold. It was in \cite{Imbr1} where Imbrie settled this question and our proof is inspired by his argument.
\end{remark} 
\begin{figure}[H]\label{A1}
\centering
\tikzset{every picture/.style={line width=0.75pt}} 
\begin{tikzpicture}[scale=1.8]
	\draw [ line width=0.5mm, fill=black!8!white, draw=black] plot [smooth cycle] coordinates {(-2,1) (-2,-1) (0,-2) (2,-2) (6,0) (4,2) (0,2) }--cycle;
	\draw [ line width=0.5mm, fill=white, draw=black] plot [smooth cycle] coordinates {(-1.5,1) (-1,-1.2) (1,-1.5) (5.5,0) (3.5, 1.8) }--cycle;
	\draw [ line width=0.5mm, fill=black!8!white, draw=black] plot [smooth cycle] coordinates {(-0.75,-1) (1,-1) (2.8,0) (1.5, 1.2) (-1,1) }--cycle;
	\draw [ line width=0.5mm, fill=white, draw=black] plot [smooth cycle] coordinates {(-0.55,-0.7) (0.75,-0.75) (1,0)  (-0.5,0.5) }--cycle;
	\draw [ line width=0.5mm, fill=white, draw=black] plot [smooth cycle] coordinates {(1.5,-0.4)  (2.5, 0.2) (1.7,0.8) (0,0.7) (1, 0.4) }--cycle;
	\begin{scope}[shift={(0,-0.3)}]
		\draw [ line width=0.5mm, pattern=dots, draw=black] plot [smooth cycle] coordinates {(0,-0.2) (0.25,-0.25) (0.5,0) (0.3,0.3) (-0.1,0.2) }--cycle;
		\draw [ line width=0.5mm, pattern=dots, draw=black] plot [smooth cycle] coordinates {(3.5,0) (3.75,0.2) (4,0.8) (3.5, 1)  }--cycle;
		\draw (3.7,0.6) node[scale=1.5] {\large $\overline{\gamma}'$};
		\draw (0.2,0) node[scale=1.5] {\large $\overline{\gamma}'$};
	\end{scope}
	\begin{scope}[shift={(1.5,0.2)}]
		\draw [ line width=0.5mm, pattern=dots, draw=black] plot [smooth cycle] coordinates {(0,-0.2) (0.25,-0.25) (0.5,0) (0.3,0.3) (-0.1,0.2) }--cycle;
		\draw (0.2,0) node[scale=1.5] {\large $\overline{\gamma}'$};
	\end{scope}
	\draw (0.9,1) node[scale=1.5] {\large $\overline{\gamma}$};
	\draw (1.1,-1.8) node[scale=1.5] {\large $\overline{\gamma}$};
\end{tikzpicture}
	\caption{To illustrate how Condition \textbf{(A1)} works, consider the figure above. In this case, the connected components of $\overline{\gamma}'$ are dotted and the connected components of $\overline{\gamma}$ are grey. One can readily see that there is a connected component of $(\overline{\gamma})^c$ that has a nonempty intersection with $V(\overline{\gamma}')$ but does not fully contain it. In order to fix such a problem, one should separate $\overline{\gamma}'$ in three different sets.}
\end{figure}
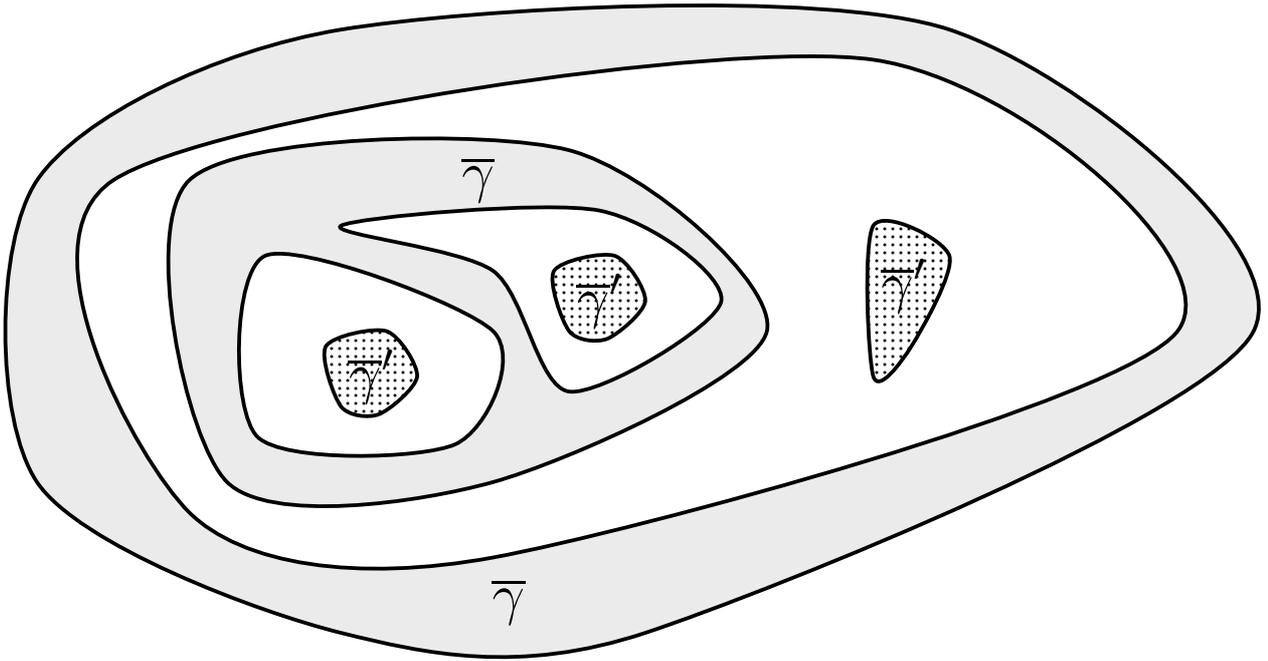

\subsection{The definition and a discussion about contours on long-range Ising models}

For a fixed configuration $\sigma$ with finite boundary $\partial \sigma$, the $(M,a,r)$-partitions will be the \textit{support of the contours},  subsets of $ \Z^d$ where every point is incorrect. The starting point for contour-based phase transition arguments for $d\geq 2$ long-range models was Park's extension of Pirogov-Sinai theory, developed in \cite{Park1, Park2}. His methods applied for the problem in consideration in this work only allowed to study long-range interactions with $\alpha>3d+1$.  At this point, although we are strongly inspired by the papers of Fr\"{o}hlich and Spencer \cite{Fro3, Fro1, Fro2} (See also \cite{Imbr1, Imbr2}) and the previous authors, we implemented slight modifications that allow us to cover all the region $\alpha > d$. 

The first point is that, differently from the original papers of Fr\"{o}hlich and Spencer, we have no arithmetic condition over the $(M,a,r)$-partitions, which means we do not ask the sum of the spins should be zero over the support of the contours. In particular, there is no constraints over the size of the contours they could have an even or odd number of points of $\Z^d$. The definition is stated only in terms of the distances among subsets of $\Z^d$.

	In order to control the entropy (the quantity) of contours, we needed to introduce a parameter $r$. It is worthwhile to stress that in the original works of Fr\"{o}hlich and Spencer \cite{Fro3, Fro1, Fro2} the parameter $r=1$.  
	
	The exponent $a$ plays an important role in our arguments.  When $\alpha$ is close to $d$, the interaction is stronger, and the contours should be far from each other. In the original papers of Fr\"{o}hlich and Spencer, $a$  is chosen as a fixed number or to belong in a finite interval as follows.  In the first paper \cite{Fro3}, for bidimensional models, we have $\frac{3}{2} < a <2$.  In the paper about the Dyson model when $d=1$ and $\alpha =2$, we have $a=\frac{3}{2}$.  In the paper of multidimensional random Schr\"{o}dinger operators,  they assumed $1 \leq a <2$, in section 4 of \cite{Fro2}. These papers use the idea of multiscale analysis to study different problems, which strongly inspired us in the definition and construction of our contours. A brief summary mentioning the multiscale methods and its power is in \cite{Simon} and at the references therein.
	
After this,  Cassandro, Ferrari, Merola, Presutti \cite{Cass} defined the contours with a more geometrical approach, using triangles. Their energy bounds for erasing a contour depend on these triangles. They choose $a = 3$ for unidimensional long-range Ising models with $\alpha \in (2-\alpha^+, 2]$, where $\alpha^+=\log(3)/\log(2) - 1$.  This is an important point; while our arguments work for any $2\leq d<\alpha$,  Littin and Picco proved that in the unidimensional case, it is impossible to produce a direct proof of the phase transition using the Peierls' argument and definition of contour in \cite{Cass}, although they prove the phase transition for the entire region with $\alpha \in (1,2]$. The papers  \cite{Cass} and \cite{LP} also assume the extra assumption $J(1) \gg 1$, which means that first neighbor interaction should be big enough, and the model simulates a short-range behavior. Recently it was proved that this extra assumption could be removed \cite{Eric2, Kim}.

 The following proposition guarantees  the existence of an $(M,a,r)$-partition for each configuration $\sigma$ with finite boundary.
\begin{proposition}\label{prop1}
	Fix real numbers $M,a,r>0$. For every $\sigma \in \Omega$ with finite boundary there is an $(M,a,r)$-partition $\Gamma(\sigma)$. Moreover, it satisfies condition \textbf{(A1)}.
\end{proposition}
\begin{proof}
	For each $x \in \Z^d$ and $n> 0$ we define an $n$-cube $C_n(x) \subset \Z^d$ as
	\begin{equation}\label{cubes}
	C_n(x) \coloneqq \left(\prod_{i=1}^d[2^{n-1}x_i-2^{n-1},2^{n-1}x_i+2^{n-1}]\right)\cap \Z^d.
	\end{equation}
	These cubes have sidelength $2^n$ and center at the point $2^{n-1} x$. For $n=0$, we establish the convention that $C_{0}(x) = x$, for any point $x \in \Z^d$. For each $\Lambda \Subset \Z^d$ and $n\geq 0$, we define $\mathscr{C}_n(\Lambda)$ as a minimal cover of $\Lambda$ by $n$-cubes. For each cover $n\geq 0$, we define the graph $G_n(\Lambda) = (v(G_n(\Lambda)),e(G_n(\Lambda)))$ by $v(G_n(\Lambda))\coloneqq\mathscr{C}_n(\Lambda)$ and $e(G_n(\Lambda))\coloneqq \{(C_n(x),C_n(y)): \dis(C_n(x),C_n(y))\leq Md^a2^{an}\}$.
	
	Note that $d2^n$ is the diameter in the $\ell^1$-norm of an $n$-cube. Let $\mathscr{G}_n(\Lambda)$ be the set of all connected components of the graph $G_n(\Lambda)$ and, for each $G \in \mathscr{G}_n(\Lambda)$, define
	\[
	\gamma_G \coloneqq \bigcup_{C_n(x) \in v(G)} (\Lambda \cap C_n(x)).
	\]
	We are ready to establish the existence of an $(M,a,r)$-partition for the boundary $\partial\sigma$ of a configuration  $\sigma$. Set $\partial \sigma_0:=\partial\sigma$ and 
	\[
	\mathscr{P}_0:=\{G \in \mathscr{G}_0(\partial \sigma_0) : |v(G)|\leq 2^r-1\}.
	\]
	Notice that this set separates all points that are distant by at least $M d^a$. Define inductively, for $n\geq 1$,  
	\[
	\mathscr{P}_n:=\{G \in \mathscr{G}_n(\partial \sigma_n): |v(G)|\leq 2^r-1\},
	\]
	where $\partial\sigma_n:=\partial\sigma_{n-1}\setminus \underset{G \in \mathscr{P}_{n-1}}{\bigcup}\gamma_G$, for $n\geq1$. Since the $n$-cubes invade the lattice, when we continue increasing $n$, there exists $N\geq 0$ such that $\partial\sigma_n = \emptyset$ for every $n \geq N$. In this case, we define $\mathscr{P}_n = \emptyset$. Let $\mathscr{P} \coloneqq \underset{n\geq 0}{\bigcup}\mathscr{P}_n$. We are going to show that the family $\Gamma(\sigma)\coloneqq \{\gamma_G: G \in \mathscr{P}\}$ is an $(M,a,r)$-partition. 
	
	In order to show that Condition \textbf{(B)} is satisfied, we will construct families of subsets, with at most $2^r-1$ elements, where Inequality \eqref{B_distance} is verified. We will write $G_n \coloneqq G_n(\partial \sigma_n)$ to simplify our notation. Take distinct $\gamma_G, \gamma_{G'} \in \Gamma(\sigma)$. There are positive integers $n,m$, with $n \leq m$, such that $G \in \mathscr{P}_n$ and $G' \in \mathscr{P}_m$. Let $G''$ be the subgraph of $G_n$ such that $v(G'')$ covers $\gamma_{G'}$ and it is minimal in the sense that all other subgraphs $G'''$ of $G_n$ satisfying this property have $|v(G''')|\geq |v(G'')|$. Thus, defining $\gamma_k = \gamma_G\cap C_n(x_k)$, for each $C_n(x_k) \in v(G)$, we have,
	\be\label{eq_distance}
	\dis(\gamma_G, C_n(z)) = \min_{1 \leq k \leq |v(G)|}\dis(\gamma_k, C_n(z)) \geq \min_{1 \leq k \leq |v(G)|} \dis(C_n(x_k),C_n(z)),
	\ee
	for each $C_n(z)\in v(G'')$. There is no edge between the subgraph $G''$ and the connected component $G$, by construction. 
	Thus, 
	\[
	\dis(C_n(x_k),C_n(z)) > Md^a 2^{an}.
	\]
	Consider, also, the sets $\gamma'_j = \gamma_{G'}\cap C_m(y_j)$, where $C_m(y_j) \in v(G')$. Then
	\be
		\dis(\gamma_G, \gamma_{G'})= \min_{1 \leq j \leq |v(G')|} \min_{\substack{C_n(z) \in v(G'') \\ C_n(z) \cap \gamma'_j \neq \emptyset}}\dis(\gamma_G,\gamma'_j \cap C_n(z)) \geq \min_{ C_n(z) \in v(G'')} \dis(\gamma_G, C_n(z)).
	\ee
	Using Inequality \eqref{eq_distance}, we arrive at the inequality $\dis(\gamma_G,\gamma_{G'})> Md^a 2^{an}$. Note that, by construction, both $\{\gamma_k\}_{1\leq k \leq |v(G)|}$ and $\{\gamma'_j\}_{1\leq j \leq |v(G')|}$ satisfy, respectively, 
	
	\[
	\max_{1\leq k \leq |v(G)|}\diam(\gamma_k) \leq d2^n \quad \text{and} \quad \max_{1\leq j \leq |v(G')|}\diam(\gamma'_j) \leq d2^m.
	\] 
	Since we assumed that $n \leq m$, we get
	\[
	\min\left\{\max_{1\leq k \leq |v(G)|}\diam(\gamma_k) ,\max_{1\leq j \leq |v(G')|}\diam(\gamma'_j)\right\}^a \leq d^a 2^{an},
	\]
	and we proved that the family $\Gamma(\sigma)$ satisfies Condition \textbf{(B)}. 
	
	In order to establish Condition \textbf{(A)}, we first note that the equality $\partial \sigma = \bigcup_{G \in \mathscr{P}}\gamma_G$ follows by construction. The elements of $\Gamma(\sigma)$ are pairwise disjoint since Inequality \eqref{B_distance} is satisfied.
	
	Eventually, let us present the proof of Condition \textbf{(A1)} for our constructed partition. Let $\gamma_G, \gamma_{G'} \in \Gamma(\sigma)$, with $V(\gamma_G)\cap V(\gamma_{G'})\neq \emptyset$. There are positive integers $n,m$ satisfying $n\leq m$ such that $G \in \mathscr{P}_n$ and $G' \in \mathscr{P}_m$. Consider $G''$, as before, the minimal subgraph of $G_n$ that covers $\gamma_{G'}$. Since $\gamma_G \cap \gamma_{G'}=\emptyset$ it holds $\gamma_G \subset (\gamma_{G'})^c$. We will show that $\gamma_G$ must be contained in only one connected component of $(\gamma_{G'})^c$. Every $n$-cube $C_n(x) \in v(G)$ cannot have a nonempty intersection with $\gamma_{G'}$, since the last one is covered by the $n$-cubes in $v(G'')$ and there is no edge between $G$ and $G''$. This is sufficient to conclude that each $n$-cube in $v(G)$ is in only one connected component of $(\gamma_{G'})^c$.
	 
	
	If $(\gamma_{G'})^c$ has only one connected component or $|v(G)|=1$, there is nothing to prove. Suppose, by contradiction, that there exist two $n$-cubes $C_n(x), C_n(x') \in v(G)$ in different connected components of $(\gamma_{G'})^c$. We claim 
	\be\label{eq1111}
	\dis(C_n(x), C_n(x')) \geq 2M d^a 2^{an}.
	\ee
	Indeed, take two points $z \in C_n(x)$ and $z' \in C_n(x')$ such that $\dis(C_n(x), C_n(x'))=|z-z'|$. Let $\lambda_{z,z'}$ be a minimal path in $\Z^d$ starting at $z$ and ending at $z'$. Note that $|\lambda_{z,z'}| = |z-z'|$. Since $C_n(x)$ and $C_n(x')$ are in different connected components of $(\gamma_{G'})^c$ there must exist $y \in \lambda_{z,z'} \cap \gamma_{G'}$. We can break $\lambda_{z,z'}$ as the union of minimal paths $\lambda_{z,y}$ and $\lambda_{y,z'}$. This fact implies
	\begin{align*}
	\dis(C_n(x), C_n(x')) &= |z-y| + |y-z'|\nonumber \\
	&\geq \min_{y' \in \gamma_{G'}}[|z-y'| + |y'-z'|] \nonumber\\
	&\geq \min_{y' \in \gamma_{G'}}[\dis(C_n(x),y') + \dis(C_n(x'),y')] \nonumber\\
	&\geq \min_{C_n(z) \in v(G'')}[\dis(C_n(x),C_n(z)) + \dis(C_n(x'),C_n(z))]\\
	& \geq 2M d^a 2^{an},
	\end{align*}
	where the last inequality is due to the fact that the subgraphs $G$ and $G''$ have no edge between them. Inequality \eqref{eq1111} is valid for any pair of $n$-cubes in different connected components of $(\gamma_{G'})^c$, thus our discussion implies that $C_n(x)$ and $C_n(x')$ are vertices of two different connected components. This cannot happen since $G$ is connected, we arrive at a contradiction. 
\end{proof}
\begin{definition}
	For $\Lambda \subset \Z^d$, we define the \emph{inner boundary} $\din \Lambda \coloneqq \{x \in \Lambda : \inf_{y \in \Lambda^c} |x-y| =1\}$,
	and the \emph{edge boundary} as $\partial\Lambda \coloneqq \{\{x,y\} \subset \Z^d: |x-y|=1, x \in \Lambda, y \in \Lambda^c\}$.
\end{definition}
\begin{remark}
	The usual isoperimetric inequality says $2d|\Lambda|^{1- \frac{1}{d}} \leq |\partial \Lambda|$. The inner boundary and the edge boundary are related by  $|\din \Lambda| \leq |\partial \Lambda| \leq 2d |\din \Lambda|$, yielding us the inequality $|\Lambda|^{1-\frac{1}{d}} \leq |\din\Lambda|$, which we will use in the rest of the paper. 
\end{remark}
 In order to define the label of a contour, we must be careful since the inner boundary of a set $\overline{\gamma}$ may have different signs; see the figure below. 
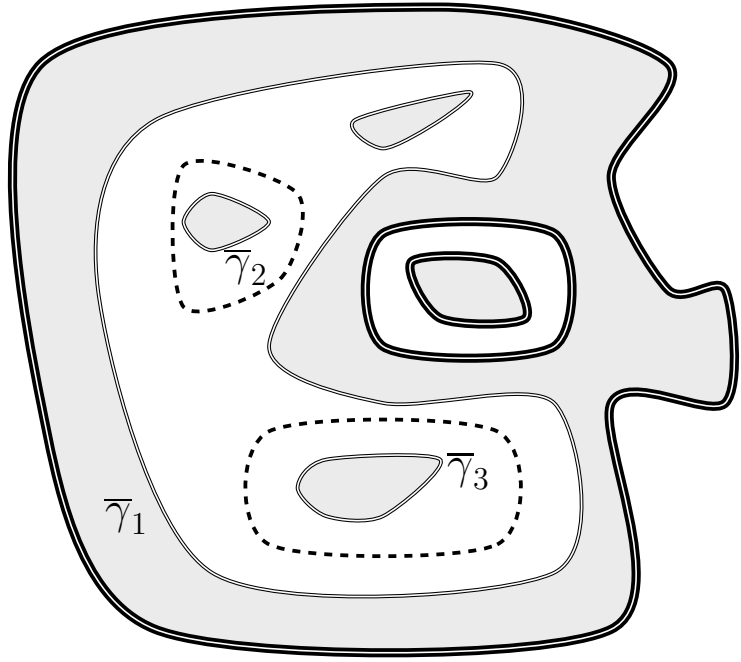
\begin{figure}[H]
	\centering
	
	\tikzset{every picture/.style={line width=0.75pt}} 
	
	\begin{tikzpicture}[scale=1.2]
		\draw [line width=0.5mm,fill=black!8!white,double] plot [smooth cycle]  coordinates {(0,-1) (4,-1) (4,1) (5,1) (5,2) (4.5,2) (4,3) (4.5,4) (2.5,4.5) (-1,4) (-1,1)}--cycle;
		\draw [line width=0.5mm,fill=white, thin, double] plot [smooth cycle] coordinates {(0.5,-0.5) (3.5,-0.5) (3.5,1) (2,1) (1, 1.5) (2,3) (3,3) (3,4) (0,3.5) (-0.5,2)}--cycle;
		\draw [line width=0.5mm,fill=white, double] plot [smooth cycle] coordinates {(2,1.5) (3.5,1.5) (3.5,2.5) (2,2.5)}--cycle;
		\draw [line width=0.5mm,fill=black!8!white, double] plot [smooth cycle] coordinates {(2.5,1.75) (3.25,1.75) (3,2.25) (2.25,2.25)}--cycle;
		\draw [line width=0.5mm,fill=black!8!white, thin, double] plot [smooth cycle] coordinates {(2,3.25) (2.5,3.5) (2.75,3.75) (1.75,3.5)}--cycle;
		
		\begin{scope}[shift={(0,0.1)}]
		\draw [ line width=0.5mm, draw=black, dashed] plot [smooth cycle] coordinates {(0.25,1.75) (1.1,2) (1.25,2.75) (0.25,3) }--cycle;
		\draw [line width=0.5mm, fill=black!8!white, thin, double] plot [smooth cycle] coordinates {(0.5,2.25) (1,2.5) (0.5,2.75) (0.25,2.5)}--cycle;
		\end{scope}
		\draw [line width=0.5mm, draw=black, dashed] plot [smooth cycle] coordinates {(1,-0.25) (3,-0.25) (3,0.75) (1,0.75) }--cycle;
		\draw [line width=0.5mm, fill=black!8!white, thin, double] plot [smooth cycle] coordinates {(1.5,0) (2,0) (2.5,0.5) (1.5,0.5) (1.25,0.25) }--cycle;
		
		\draw (-0.25,0) node[scale=1] {\Large $\overline{\gamma}_1$};
		\draw (0.8,2.2) node[scale=1] {\Large $\overline{\gamma}_2$};
		\draw (2.75,0.4) node[scale=1] {\Large $\overline{\gamma}_3$};
		
	\end{tikzpicture}
	\caption{An example of $\Gamma(\sigma)=\{\overline{\gamma}_1,\overline{\gamma}_2, \overline{\gamma}_3\}$, with $\overline{\gamma}_1$ having regions in the  inner boundary with different signs. In the figure, the grey region are for incorrect points, the thin and thick border corresponds to, respectively, $+1$ and $-1$ labels.}
\end{figure}

For a set $\overline{\gamma} \in \Gamma(\sigma)$, let $\overline{\gamma}^{(1)},\dots,\overline{\gamma}^{(n)}$ be its connected components. To define the label of $\overline{\gamma}$, we must introduce the following concept. A connected component $\overline{\gamma}^{(k)}$ of $\overline{\gamma}$ is called \emph{external} if for any other connected component $\overline{\gamma}^{(k')}$ with $V(\overline{\gamma}^{(k')})\cap V(\overline{\gamma}^{(k)}) \neq \emptyset$, we have $V(\overline{\gamma}^{(k')})\subset V(\overline{\gamma}^{(k)})$.
\begin{lemma}
	Any configuration $\sigma$ with finite boundary is constant on $\din V(\overline{\gamma})$, for each $\overline{\gamma} \in \Gamma(\sigma)$. 
\end{lemma}
\begin{proof}
	 Note that $V(\overline{\gamma}) $ is the union of $V(\overline{\gamma}^{(k)})$ for all its external connected components $ \overline{\gamma}^{(k)}$. Suppose there are $\overline{\gamma}^{(k)},\overline{\gamma}^{(j)}$ external connected components of $\overline{\gamma}$ such that the sign of $\sigma$ on  $\din V(\overline{\gamma}^{(k)})$ is different from the sign on $\din V(\overline{\gamma}^{(j)})$. Then, the configuration $\sigma$ must change sign inside $V(\overline{\gamma}^{(k)})^c\cap V(\overline{\gamma}^{(j)})^c$. Since $\sigma$ is constant outside some finite set $\Lambda$, either $\overline{\gamma}^{(k)}$ or $\overline{\gamma}^{(j)}$ must be surrounded by a different region of incorrect points, let us call it $\overline{\gamma}^{(l)}$. We can assume that the connected component surrounded by $\overline{\gamma}^{(l)}$ is $\overline{\gamma}^{(k)}$. The set $\overline{\gamma}^{(l)}$ cannot be a connected component  of $\overline{\gamma}$, otherwise the set $\overline{\gamma}^{(k)}$ would not be external. If $\overline{\gamma}^{(l)}$ is a connected component of an other element $\overline{\gamma}' \in \Gamma(\sigma)$, then $\overline{\gamma}$ have a nonempty intersection with at least two connected components of $(\overline{\gamma}')^c$, contradiction to Condition \textbf{(A)}. 
\end{proof}
The \textit{label} of $\overline{\gamma}$ is defined as the function $\lab_{\overline{\gamma}} :\{(\overline{\gamma})^{(0)}, \I(\overline{\gamma})^{(1)}\dots, \I(\overline{\gamma})^{(n)}\} \rightarrow \{-1,+1\}$ defined as: $\lab_{\overline{\gamma}}(\I(\overline{\gamma})^{(k)})$ is the sign of the configuration $\sigma$ in $\din V(\I(\overline{\gamma})^{(k)})$, for $k\geq 1$, and $\lab_{\overline{\gamma}}((\overline{\gamma})^{(0)})$ is the sign of $\sigma$ in $\din V(\overline{\gamma})$.
\begin{definition}
Given a configuration $\sigma$ with finite boundary, its \emph{contours} $\gamma$ are pairs $(\overline{\gamma},\lab_{\overline{\gamma}})$,  where $\overline{\gamma} \in \Gamma(\sigma)$ and $\lab_{\overline{\gamma}}$ is the label function defined previously. The \emph{support of the contour}  $\gamma$ is defined as $\Sp(\gamma)\coloneqq \overline{\gamma}$ and its \emph{size} is given by $|\gamma| \coloneqq |\Sp(\gamma)|$.
\end{definition}

Another important concept for our analysis of phase transition is the interior of a contour. The following sets will be useful 
\[
\I_\pm(\gamma) = \bigcup_{\substack{k \geq 1, \\ \lab_{\overline{\gamma}}(\I(\Sp(\gamma))^{(k)})=\pm 1}}\I(\Sp(\gamma))^{(k)} , \;\;\;
\I(\gamma) = \I_+(\gamma) \cup \I_-(\gamma), \;\;\;
V(\gamma) = \Sp(\gamma) \cup \I(\gamma),
\]
where $\I(\Sp(\gamma))^{(k)}$ are the connected components of $\I(\Sp(\gamma))$. Notice that the interior of contours in Pirogov-Sinai theory are at most unions of simple connected sets. In our case, they are only connected, i.e., they may have holes. One feature that our contours share with Pirogov-Sinai theory is the abscence of a bijective correspondence between contours and configurations (see section 7.2.6 of \cite{Vel}, and \cite{Park1}). Usually there is more than one configuration giving the same boundary set. Also, it is not true that for all families of contours $\Gamma \coloneqq \{\gamma_1, \gamma_2, \dots, \gamma_n\}$ there is a configuration $\sigma$ whose contours are exactly $\Gamma$. This happens because they may not form an $(M,a,r)$-partition and, even if this is the case, their labels may not be compatible. When such a configuration exists, we say that the family of contours $\Gamma$ is \emph{compatible}. 

\begin{figure}[H]
	\centering
	
	\tikzset{every picture/.style={line width=0.75pt}} 
	
	\begin{tikzpicture}[scale=0.5]
		\draw [line width=0.5mm,fill=black!8!white,  thin, double] plot [smooth cycle] coordinates {(0,-1) (4,-1) (4,1) (2.5,2) (2,3.5) (-1,4) (-1,1)}--cycle;
		\draw [line width=0.5mm,fill=white, double] plot [smooth cycle] coordinates {(1.5,0) (2,2) (0.5,2.5) (0,1)}--cycle;
		\draw [line width=0.5mm,fill=black!8!white, double] plot [smooth cycle] coordinates {(1,1) (1.5,1.5) (1,2) (0.5,1.5)}--cycle;
		
		\draw [line width=0.5mm,fill=black!8!white,  thin, double] plot [smooth cycle] coordinates {(5,4) (7,4) (7,8) (8,8) (8,10) (6,10) (6,11) (3,11) (3, 7)}--cycle;
		
		\draw [line width=0.5mm,fill=black!8!white,  thin, double] plot [smooth cycle] coordinates {(22,4) (24,4) (24,8) (25,8) (25,10) (23,10) (23,13) (19,13) (18,10) (16,6) (16,5)}--cycle;
		\draw [line width=0.5mm,fill=white, double] plot [smooth cycle] coordinates {(17,6) (23,5) (23,7) (20,8) }--cycle;
		\draw [line width=0.5mm,fill=black!8!white,  thin, double] plot [smooth cycle] coordinates {(20,6.25) (21,6.25) (21,6.9) (20,6.9) }--cycle;
		\draw [line width=0.5mm,fill=white, thin, double] plot [smooth cycle] coordinates {(19,10) 
			(22,10) (22,12) (20,12) }--cycle;
		\draw [line width=0.5mm,fill=black!8!white, double] plot [smooth cycle] coordinates {(20.25,10.75) (21.25,10.75) (21.25,11.4) (20.25,11.4) }--cycle;
		
		\draw (0,0) node[scale=0.75] {\Huge$\gamma_1$};
		\draw (1,1.5) node[scale=0.75] {\large$\gamma_2$};
		\draw (5,7) node[scale=0.75] {\Huge$\gamma_3$};
		\draw (22,8.5) node[scale=0.75] {\Huge$\gamma_4$};
		\draw (20.5,6.5) node[scale=0.75] {\Large$\gamma_5$};
		\draw (20.8,11) node[scale=0.75] {\Large$\gamma_6$};
	\end{tikzpicture}
	\caption{Above we have two situations where incompatibility happens. In the first case, we have $\gamma_1$ and $\gamma_2$ two contours that are close, thus they should not be separated. In the case $\gamma_4,\gamma_5,\gamma_6$ we have the usual problem of labels not matching.}
\end{figure}
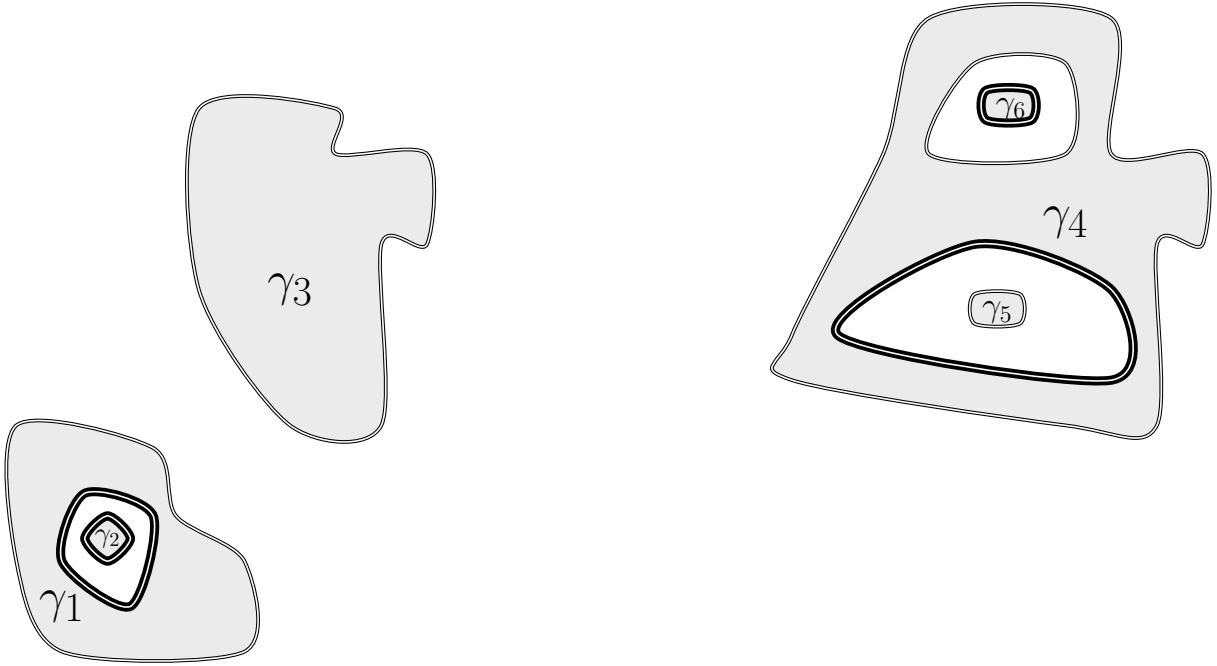

\subsection{Entropy Bounds}

The proofs in this section are highly inspired by Section 4 of Fr\"{o}hlich-Spencer \cite{Fro2}, the one-dimensional case studied in Fr\"{o}hlich-Spencer \cite{Fro1} and by Cassandro, Ferrari, Merola, Presutti \cite{Cass}. We say that a contour $\gamma$ in a family $\Gamma$ is \emph{external} if its external connected components are not contained in any other $V(\gamma')$, for $\gamma' \in \Gamma\setminus\{\gamma\}$. For each $\Lambda \subset \Z^d$, let us define the set of all external compatible families of contours $\Gamma$ with label $\pm$ contained in $\Lambda$ by
\[
\mathcal{E}^\pm_\Lambda \coloneqq\{\Gamma= \{\gamma_1, \ldots, \gamma_n\}: \Gamma \text{ is compatible, } \gamma_i \text{ is external}, \lab(\gamma_i)=\pm1, V(\Gamma) \subset \Lambda\},
\]
where $V(\Gamma)=\bigcup_{1 \leq i\leq n}V(\gamma_i)$ and $\lab(\gamma_i)$ is the value of the corresponding configuration compatible configuration in a point $x\in \din V(\gamma_i)$. When we write $\gamma \in \mathcal{E}^\pm_\Lambda$ we mean $\{\gamma\} \in \mathcal{E}^\pm_\Lambda$. We also write $\mathcal{E}_{\Z^d}^\pm\coloneqq \mathcal{E}^\pm$. To hold a Peierls-type argument, it is important to control the number of contours with a given fixed size. Hence, we need to find an upper bound of the number of contours with fixed size $|\gamma|$ containing a given point, meaning $0\in V(\gamma)$. To do so, we will need some auxiliary results. 

Since a contour does not need to be connected, the counting argument is different from the short-range case. Let us denote by $\mathscr{C}_n$ an arbitrary collection of $n$-cubes. For $n,m\ge 0$ with $n\le m$, we say that  
$\mathscr{C}_{n}$ is \emph{subordinated} to $\mathscr{C}_{m}$, denoted by $\mathscr{C}_n \preceq \mathscr{C}_m$, if $\mathscr{C}_m$ is a minimal cover of $\cup_{C_n(x)\in \mathscr{C}_n}C_n(x)$.
For each $n,m\geq 1$, with $n\leq m$, define
\[
N(\mathscr{C}_m, V_n) \coloneqq |\{\mathscr{C}_{n}: \mathscr{C}_n \preceq \mathscr{C}_m, |\mathscr{C}_n|=V_n\}|
\]
be the number of collections of $n$-cubes $\mathscr{C}_n$ subordinated to a given $\mathscr{C}_m$ such that $|\mathscr{C}_n| = V_n$. 
The following two propositions are straightforward generalizations of Theorem 4.2 and Proposition 4.3 of \cite{Fro2}. The original proofs correspond to our case when $r=1$, and are included here for the benefit of the reader. 

\begin{proposition}\label{appB:prop1}
	Let $r,n\geq 1$ be integers, and $\mathscr{C}_{rn}$ be a fixed collection of $rn$-cubes. Then there exists a constant $c\coloneqq c(d,r)>0$ such that
	\be\label{eq_appB}
	N(\mathscr{C}_{rn}, V_{r(n-1)}) \leq e^{c V_{r(n-1)}}.
	\ee
\end{proposition}
\begin{proof}
	For each $rn$-cube $C_{rn}(x) \in \mathscr{C}_{rn}$, let $N_{C_{rn}(x)}$ be the number of cubes in a collection of $r(n-1)$-cubes $\mathscr{C}_{r(n-1)}$ that are covered by $C_{rn}(x)$. Fix $(n_{C_{rn}(x)})_{C_{rn}(x) \in \mathscr{C}_{rn}}$, with $n_{C_{rn}(x)}\geq 1$, an integer solution to the inequality
	\be\label{100}
	V_{r(n-1)}\leq \sum_{C_{rn}(x) \in \mathscr{C}_{rn}}n_{C_{rn}(x)}
	\leq 2dV_{r(n-1)},
	\ee
	and define $N(\mathscr{C}_{rn}, V_{r(n-1)}, (n_{C_{rn}(x)})_{C_{rn}(x) \in \mathscr{C}_{rn}})$ to be the number of collections $\mathscr{C}_{r(n-1)}$ of $r(n-1)$-cubes subordinated to $\mathscr{C}_{rn}$ such that $|\mathscr{C}_{r(n-1)}|=V_{r(n-1)}$ and $N_{C_{rn}(x)}=n_{C_{rn}(x)}$ for each $C_{rn}(x) \in \mathscr{C}_{rn}$. We get
	\[
	N(\mathscr{C}_{rn}, V_{r(n-1)}) = \sum_{(n_{C_{rn}(x)})_{C_{rn}(x) \in \mathscr{C}_{rn}}}N(\mathscr{C}_{rn}, V_{r(n-1)}, (n_{C_{rn}(x)})_{C_{rn}(x) \in \mathscr{C}_{rn}}).
	\]
	The number of positions that a $r(n-1)$-cube can be arranged inside a $rn$-cube is at most $(2^{r+1}-1)^d$. Since the number of combinations that $n_{C_{rn}(x)} r(n-1)$-cubes can be arranged in the cube $C_{rn}(x)$ is at most $\binom{(2^{r+1}-1)^d}{n_{C_{rn}(x)}}$ we get
	\[
	N(\mathscr{C}_{rn}, V_{r(n-1)}, (n_{C_{rn}(x)})_{C_{rn}(x) \in \mathscr{C}_{rn}}) \leq \prod_{C_{rn}(x) \in \mathscr{C}_{rn}}\binom{(2^{r+1}-1)^d}{n_{C_{rn}(x)}}.
	\]
	The number of solutions to \eqref{100} is bounded by $2^{2dV_{r(n-1)}+1}$, concluding that Inequality (\ref{eq_appB}) holds for $c = (2d+1)\log(2)+ d\log(2^{r+1}-1)$.
\end{proof}
 Given a subset $\Lambda \Subset \Z^d$ and integers $r\geq 1$ and $n\geq 0$, define the \emph{total volume} by
\be
V_r(\Lambda) \coloneqq \sum_{n=0}^{n_r(\Lambda)} |\mathscr{C}_{rn}(\Lambda)|,
\ee  
where $n_r(\Lambda) = \lceil \log_{2^r}(\diam(\Lambda))\rceil$ and $\mathscr{C}_{rn}(\Lambda)$ is a \emph{minimal cover} of $\Lambda$ with $rn$-cubes $C_{rn}(x)$ defined in (\ref{cubes}). Observe that $|\mathscr{C}_{0}(\Lambda)|=|\Lambda|$. 
Let $V\geq 1$ be a positive integer and $\mathscr{F}_V$ be the set defined by
\be
\mathscr{F}_V \coloneqq \{\Lambda \Subset \Z^d : V_r(\Lambda) =V, 0 \in \Lambda\}.
\ee
By using Proposition \ref{appB:prop1}, let us show that the number of elements in $\mathscr{F}_V$ is exponentially bounded by $V$.
\begin{proposition}\label{appB:prop2}
	There exists $b\coloneqq b(d,r) >0$ such that,
	\be\label{bound_F_V}
	|\mathscr{F}_V| \leq e^{bV}.
	\ee
\end{proposition}
\begin{proof}
	For each $\Lambda \in \mathscr{F}_V$, there is a family of minimal covers $(\mathscr{C}_{rn}(\Lambda))_{0\leq n \leq n_r(\Lambda)}$ characterizes in a unique kind of way the subset $\Lambda \in \mathscr{F}_V$ since the minimal cover $\mathscr{C}_0(\Lambda)$ is a cover by points in $\mathbb{Z}^d$. Moreover, the minimal covers $\mathscr{C}_{rn}(\Lambda)$ can always be chosen in a way that $\mathscr{C}_{rn_1}$ is subordinated to $\mathscr{C}_{rn_2}$ whenever $n_1 \leq n_2$, since, in order to compute the total volume $V_r(\Lambda)$, we only need to know the size of each minimal cover $\mathscr{C}_{rn}(\Lambda)$. 
	Fix $(V_{rn})_{0 \leq n \leq n_r(\Lambda)-1}$ a solution to the equation
	\be\label{constraint}
	\sum_{n=0}^{n_r(\Lambda)-1} V_{rn}= V-1. 
	\ee
	We can estimate $|\mathscr{F}_V|$ by counting the number of families $(\mathscr{C}_{rn}(\Lambda))_{0 \leq n \leq n_r(\Lambda)}$ where the last cover $\mathscr{C}_{rn_r(\Lambda)}$ consists of a unique cube $C_{rn_r(\Lambda)}(x)$ containing $0$. Let $\mathscr{F}_{V,m}=\{\Lambda \in \mathscr{F}_V:  n_r(\Lambda)=m\}.$ Then,
	\[
		|\mathscr{F}_V| \leq \sum_{m=1}^{V}|\mathscr{F}_{V,m}|,
	\]
	since $n_r(\Lambda)\leq V_r(\Lambda)=V$. Now,
	\begin{align*}
	|\mathscr{F}_{V,m}| &\leq \sum_{(V_{rn})_{0 \leq n \leq m-1}}|\{(\mathscr{C}_{rn})_{0 \leq n \leq m}: \mathscr{C}_{rn}\preceq \mathscr{C}_{r(n+1)}, |\mathscr{C}_{rn}|=V_{rn}, \mathscr{C}_{rm}=\{C_{rm}(x)\}, 0 \in C_{rm}(x) \}| \\
	&=\sum_{(V_{rn})_{0 \leq n \leq m-1}}\sum_{C_{rm}(x) \ni 0}\sum_{\substack{\mathscr{C}_{r(m-1)} \\ |\mathscr{C}_{r(m-1)}|=V_{r(m-1)}\\ \mathscr{C}_{r(m-1)} \preceq \mathscr{C}_{rm}}}\ldots \sum_{\substack{\mathscr{C}_{r} \\ |\mathscr{C}_r|=V_r\\ \mathscr{C}_r \preceq \mathscr{C}_{2r}}} N(\mathscr{C}_r, V_0).
	\end{align*}
Iterating Inequality \eqref{eq_appB} we get 
\be\label{101}
	|\mathscr{F}_{V,m}| \leq  |\{C_{rm}(x): 0\in C_{rm}(x)\}| \sum_{(V_{rn})_{0 \leq n \leq m-1}} \prod_{n=0}^{m-1}e^{c V_{rn}}.
\ee
	 We have at most $2^V$ solutions for Equation \eqref{constraint}, thus Inequality \eqref{101} together with the fact that the number of $rm$-cubes containing $0$ is bounded by $3^d$ yield us
	\[
	|\mathscr{F}_V| \leq 3^d V 2^V e^{c V}.
	\]
Therefore, Inequality (\ref{bound_F_V}) holds for $b = d\log(3) + \log(2)+ c+1$. 
\end{proof}

We are able to prove Proposition \ref{importprop} once we show that a fixed configuration $\sigma$ with $\Gamma(\sigma)=\{\gamma\}$ and a fixed volume $|\gamma|=m$ implies that the total volume $V_r(\Sp(\gamma))$ is finite.  We need the following auxiliary result about graphs, which is a generalization of Claim 4.2 in \cite{KP}.   
\begin{proposition}\label{proptree}
	Let $k\ge 1$ and $G$ be a finite, non-empty, connected simple graph. Then, $G$ can be covered by $\lceil |v(G)|/k \rceil$ connected subgraphs of size at most $2k$.
\end{proposition}

\begin{proof}
	Since we can always consider a spanning tree from a connected graph $G$, it is sufficient to prove the proposition when $G$ is a tree. If either $k=1$ or $|v(G)|\leq 2k$, our statement is trivially true, so we suppose $k\geq 2$ and $|v(G)|\geq 2k+1$, and proceed by induction on $\lceil |v(G)|/k \rceil$.
	
	Choose a vertex $r \in v(G)$ to be our root. For every vertex $u$ of $G$ let $\text{dep}(u)$ be \emph{the depth} of the vertex $u$, i.e., the distance in the graph between $r$ and $u$. We say that a vertex $w$ is a \emph{descendant} of $u$ if there is a path $u_1=w, u_2, \dots, u_{n-1}, u_n=u$ in $G$ with $\text{dep}(u_i)> \text{dep}(u)$, for all $1\leq i \leq n-1$ and let $\text{desc}(u)$ be the number of descendants of $u$. Take a vertex $u^*$ from $\{u \in v(G): \text{desc}(u)\geq k\}$, that is not empty since $\text{desc}(r)\geq 2k$, with highest depth, i.e., such that for any other $u \in v(G)$ with $\text{desc}(u)\geq k$ we have $\text{dep}(u^*)\geq \text{dep}(u)$. Let $u_1,\dots,u_t$ be the children of the vertex $u^*$, and define $a_i = \text{desc}(u_i)+1$. By definition of $u^*$, we have that $a_i\leq k$, for $1\leq i \leq t$, and $a_1+ \dots +a_t \geq k$. Hence, there must be an $1\leq s \leq t$ for which $k \leq a_1+ \dots +a_s < 2k$. Therefore, we can consider the subtree $T$ whose vertex set is composed by $u^*, u_1, \dots, u_s$ and their descendants. By construction, it holds $k+1 \leq |v(T)|\leq 2k$. 
	The induced subgraph $H \coloneqq G[(v(G)\setminus v(T))\cup \{u^*\}]$ is connected and satisfies $|v(H)|\leq |v(G)|-k$. Using the induction hypothesis, $H$ can be covered by $\lceil|v(H)|/k\rceil\leq \lceil|v(G)|/k\rceil-1$ connected subgraphs. Adding $T$ to this cover completes the proof. 
\end{proof}

  The importance of the choice $r=\lceil \log_2(a+1)\rceil +d+1$ in the $(M,a,r)$-partition will be seen in the next proposition, where we will show that the total volume can be bounded by the size of the contour.

\begin{proposition}\label{prop4}
	There exists a constant $\kappa\coloneqq \kappa(d,M,r) >0$ such that, for any contour $\gamma \in \mathcal{E}_\Lambda^\pm$,
	\be\label{eq100}
	V_r(\Sp(\gamma)) \leq \kappa|\gamma|.
	\ee 
\end{proposition}
\begin{proof}
	Define $g:\mathbb{N} \rightarrow \Z$ by
	\be\label{prop4:eq1}
	g(n) = \left\lfloor \frac{n-2-\log_{2^r}(2Md^a)}{a}\right\rfloor.
	\ee
	We are going to prove 
	\be\label{prop4:eq2}
	|\mathscr{C}_{rn}(\Sp(\gamma))| \leq \frac{1}{2^{r-d-1}}|\mathscr{C}_{rg(n)}(\Sp(\gamma))|,
	\ee
	whenever $g(n)>0$, and either the graph $G_{rg(n)}(\Sp(\gamma))$ defined in Proposition \ref{prop1} has at least two connected components or it has only one connected component with at least $2^r$ vertices. Remember that $\mathscr{G}_{rg(n)}(\Sp(\gamma))$ is the set of all connected components of the graph $G_{rg(n)}(\Sp(\gamma))$. Note that
	\be\label{eq9}
	|\mathscr{C}_{rg(n)}(\Sp(\gamma))| = 2^r\sum_{G \in \mathscr{G}_{rg(n)}(\Sp(\gamma))} \frac{|v(G)|}{2^r}.
	\ee
	
	 Proposition \ref{proptree} states that we can cover the vertex set $v(G)$ with a family of connected graphs $G_i$ with $1 \leq i \leq \lceil |v(G)|/2^r \rceil$ and $ |v(G_i)|\leq 2^{r+1}$. Using the inequality 
	\[
	\diam(\Lambda \cup \Lambda') \leq \diam(\Lambda) + \diam(\Lambda') + \dis(\Lambda,\Lambda'), \quad \text{for all } \Lambda, \Lambda' \Subset \Z^d,
	\]
	and the fact that we can always extract a vertex of a finite connected graph in a way that the induced subgraph is still connected, by removing a leaf of a spanning tree, we can bound the diameter of $B_{G_i} = \bigcup_{C_{rg(n)}(x)\in v(G_i)}C_{rg(n)}(x)$ by
	\begin{align}\label{eq10}
		\diam(B_{G_i}) &\leq \sum_{C_{rg(n)}(x) \in v(G_i)}\diam(C_{rg(n)}(x)) + |v(G_i)|Md^a 2^{arg(n)} \nonumber \\
		&\leq d2^{r(g(n)+1)+1}+ Md^a2^{r(ag(n)+1)+1} \nonumber \\
		& \leq 2^{rn}.
	\end{align}
	The third inequality holds since $M,a,r \geq  1$. Therefore, each graph $G_i$ can be covered by one cube of side length $2^{rn}$ with center in $\Z^d$. We claim that every cube of side length $2^{rn}$ with an arbitrary center in $\Z^d$ can be covered by at most $2^d$ $rn$-cubes $C_{rn}(x)$. 
	Note that it is enough to consider the case where the cube has the form 
	\[
	\prod_{i=1}^d[q_i- 2^{rn-1}, q_i+ 2^{rn-1}])\cap \Z^d,
	\] 
	where $q_i \in \{0,1,\dots, 2^{rn-1}-1\}$, for $1 \leq i \leq d$.
	 It is easy to see that
	\[
	[q_i- 2^{rn-1}, q_i+ 2^{rn-1}] \subset [-2^{rn-1}, 2^{rn-1}]\cup [0, 2^{rn}].
	\] 
	Taking products for all $1\leq i \leq d$, it concludes our claim. This reasoning allows us to conclude that the maximum number of $rn$-cubes required to cover each connected component $G$ of $G_{rg(n)}(\Sp(\gamma))$ is at most $2^d\lceil |v(G)|/2^r \rceil$, yielding us  
	\be\label{eq11}
	 |\mathscr{C}_{rn}(\Sp(\gamma))| \leq \sum_{G \in \mathscr{G}_{rg(n)}(\Sp(\gamma))}\left|\mathscr{C}_{rn}\left(\bigcup_{1 \leq i \leq \lceil |v(G)|/2^r\rceil}B_{G_i}\right)\right| \leq \sum_{G \in \mathscr{G}_{rg(n)}(\Sp(\gamma))} 2^d\left\lceil\frac{|v(G)|}{2^r}\right\rceil.
	\ee
	
If $|\mathscr{G}_{rg(n)}(\Sp(\gamma))|\geq 2$, since $\gamma \in \mathcal{E}^\pm_\Lambda$, each connected component $G \in \mathscr{G}_{rg(n)}(\Sp(\gamma))$ satisfies $|v(G)|\geq 2^r$. Indeed, if $|v(G)|\leq 2^r-1$, by our construction in Proposition \ref{prop1} the set $v(G)$ would be separated into another element of the $(M,a,r)$-partition. By hypothesis, if $|\mathscr{G}_{rg(n)}(\Sp(\gamma))|=1$ it already satisfies $|v(G_{rg(n)})|\geq 2^r$. 
 Together with 
 \[
 \frac{1}{2}\left\lceil\frac{|v(G)|}{2^r}\right\rceil  \leq \frac{|v(G)|}{2^r},
 \]
 Equality \eqref{eq9} and Inequality \eqref{eq11} yield 
	\[
	|\mathscr{C}_{rn}(\Sp(\gamma))| \leq 2^{d+1}\sum_{G \in \mathscr{G}_{rg(n)}(\Sp(\gamma))}\frac{|v(G)|}{2^r} = \frac{2^{d+1}}{2^r}|\mathscr{C}_{rg(n)}(\Sp(\gamma))|.
	\]
	
	So, Inequality \eqref{prop4:eq2} is proved. Let us define two auxiliary quantities
	\[
	l_1(n) \coloneqq \max\{m \geq 0: g^m(n) \geq 0\} \quad \text{and} \quad l_2(n) \coloneqq \max\{m \geq 0: |\mathscr{G}_{rg^m(n)}(\Sp(\gamma))|=1, |v(G_{rg^m(n)})|\leq 2^r-1\}.
	\]
	For the set $\mathscr{G}_{rg^m(n)}(\Sp(\gamma))$ to be well defined, we must have $g^m(n) \geq 0$, thus $l_2(n) \leq l_1(n)$. Moreover, knowing that $|\mathscr{C}_n(\Lambda)| \leq |\Lambda|$ for any $n \geq 0$,
	\be\label{prop4:eq4}
	|\mathscr{C}_{rn}(\Sp(\gamma))| \leq  |\mathscr{C}_{rg^{l_2(n)}(n)}(\Sp(\gamma))| \leq \frac{1}{2^{(r-d-1)(l_1(n)-l_2(n))}}|\gamma|.
	\ee
	
	 We claim 
	\be\label{eq13}
	l_1(n) \geq \begin{cases}
		0 & \text{ if }0 \leq n \leq n_0, \\
		\left\lfloor \frac{\log_2 (n)-\log_2(n_0)}{\log_2(a)} \right\rfloor & \text{ if } n> n_0,
	\end{cases}
	\ee
	where 
	$n_0= (a+2+\log_{2^r}(2Md^a))(a-1)^{-1}$.
	 The first bound is trivial. Let $n > n_0$ and consider the function 
	\[
	\tilde{g}(n) = \frac{n-a-2-\log_{2^r}(2Md^a)}{a}.
	\]
	From the fact that $g(n) \geq \tilde{g}(n)$ and both functions are increasing, we have
	\be\label{ineq:gtilde}
	g^m(n) \geq \tilde{g}^m(n), \quad \text{ for all } m \geq 1,
	\ee
	which implies $\max_{\tilde{g}^m(n) \geq 0} m \leq l_1(n)$. Thus, we need to compute a lower bound for $\max_{\tilde{g}^m(n)\geq 0} m$. Since
    \[
	\tilde{g}^m(n) = \frac{n}{a^m} - b\left(\frac{a^m-1}{a^{m-1}(a - 1)}\right),
	\] 
	 is sufficient to have
	\be\label{prop4:eq5}
	\frac{n}{a^m}  > \frac{ab}{a- 1},
	\ee
	where $b=(a+2+\log_{2^r}(2Md^a))a^{-1}$.
	We get the desired bound after taking the 
	logarithm with respect to base two in both sides of Inequality \eqref{prop4:eq5}. To finish our calculation, we will analyse two cases depending if $l_2(n)$ is zero or not. First, let us consider the case where $l_2(n)=0$. Using Inequality \eqref{prop4:eq4} we get 
	\be
	V_r(\Sp(\gamma)) \leq |\gamma|\sum_{n=0}^\infty\frac{1}{2^{(r-d-1)l_1(n)}}.
	\ee
	To finish, notice that equation above can be bounded in the following way
	\[
	\sum_{n=0}^\infty\frac{1}{2^{(r-d-1)l_1(n)}} \leq n_0+ 1+2^{r-d-1}n_0^{\frac{r-d-1}{\log_2(a)}} \zeta\left(\frac{r-d-1}{\log_2(a)}\right).
	\] 	
	 Now consider the case $l_2(n)\neq 0$. A similiar bound as \eqref{eq10}, the fact that $\mathscr{C}_{rg^m(n)}(\Sp(\gamma))$ is a cover for the set $\Sp(\gamma)$ and assuming that $
    m \in \{k\geq 0: |\mathscr{G}_{rg^k(n)}|=1,|v(G_{rg^k(n)})|\leq 2^r-1\}$ yields
	\begin{align*}
	\diam(\Sp(\gamma)) \leq \diam(B_{G_{rg^m(n)}(\Sp(\gamma))}) \leq (d2^{rg^m(n)}+ Md^a 2^{arg^m(n)})|v(G_{rg^m(n)}(\Sp(\gamma)))|
	\leq 2Md^a 2^{arg^m(n)+r}.
	\end{align*}
	The inequality above yields
	\[
	\log_{2^r}(\diam(\Sp(\gamma))) \leq \log_{2^r}(2Md^a)+ ag^m(n)+1 \leq \log_{2^r}(2Md^a)+ \frac{n}{a^{m-1}}+1.
	\]
	Let us assume that $\diam(\Sp(\gamma))> 2^{2r+1}Md^a$. Isolating the term depending on $m$ and taking 
	the logarithm with respect to base two in both sides of the inequality above, it gives us
	\[
	 m \leq 1+ \frac{\log_2(n) - \log_2(\log_{2^r}(\diam(\Sp(\gamma)))-\log_{2^r}(2Md^a)-1)}{\log_2(a)}.
	\] 
	The inequality above is valid for all $m \in \{k\geq 0: |\mathscr{G}_{rg^k(n)}|=1, |v(G_{rg^k(n)})|\leq 2^r-1\}$. Thus, together with the lower bound \eqref{eq13}, we get for $n>n_0$
	\[
	l_1(n)-l_2(n) \geq \frac{ \log_2(\log_{2^r}(\diam(\Sp(\gamma)))-\log_{2^r}(2Md^a)-1)-\log_2(n_0)}{\log_2(a)}-2.
	\]
	Inequality \eqref{prop4:eq4} together with the inequality above yields  
	\begin{align}
	V_r(\Sp(\gamma)) &\leq (n_0 +1) |\gamma|+ |\gamma|\frac{2^{2(r-d-1)}n_0^{\frac{r-d-1}{\log_2(a)}}(n_r(\Sp(\gamma))-n_0)}{(\log_{2^r}(\diam(\Sp(\gamma)))-\log_{2^r}(2Md^a)-1)^{\frac{r-d-1}{\log_2(a)}}} \nonumber \\
	&\leq (n_0 + 1+ 2^{2(r-d-1)}n_0^{\frac{r-d-1}{\log_2(a)}}(2+\log_{2^r}(2Md^a)))|\gamma|,
	\end{align}
	where the last inequality is due to the fact that 
	$x/(x-w)\leq 1+w$ for any constant $x \geq w+1$ and the fact that $M\geq 1$ implies that $n_0 \geq 1$.
	 If $\diam(\Sp(\gamma))\leq 2^{2r+1}Md^a$, we have
	\be
	V_r(\Sp(\gamma))\leq (n_r(\Sp(\gamma))+1)|\gamma| \leq (4 + \log_{2^r}(2Md^a))|\gamma|.
	\ee
	Taking $\kappa = \max\{4 + \log_{2^r}(2Md^a), n_0 + 1+ 2^{2(r-d-1)}n_0^{\frac{r-d-1}{\log_2(a)}}(2+\log_{2^r}(2Md^a)), n_0+ 1+2^{r-d-1}n_0^{\frac{r-d-1}{\log_2(a)}} \zeta\left(\frac{r-d-1}{\log_2(a)}\right)\}$ concludes the desired result.
\newline
\end{proof}


We are ready to show Proposition 3.14 that bounds exponentially the number of contours containing the origin with a fixed size.

\begin{proposition}\label{importprop}
	Let $m\geq 1$, $d\ge 2$, and let the set $\mathcal{C}_0(m)$ be defined as
	\[
	\mathcal{C}_0(m) \coloneqq \{\Sp(\gamma)\Subset \Z^d:\gamma \in \mathcal{E}^-, 0 \in V(\gamma), |\gamma|=m\}.
	\]
There exists $c_1\coloneqq c_1(d,M,r)>0$ such that
	\[
	|\mathcal{C}_0(m)| \leq e^{c_1 m}.
	\]
\end{proposition}

\begin{proof}
	For a given contour $\gamma \in \mathcal{E}^-$, define the set $\mathcal{C}_\gamma$ by
	\be
	\mathcal{C}_\gamma \coloneqq \{\Sp(\gamma') \in \mathcal{C}_0(m): \exists \; x \in \Z^d \text{ s.t. } \Sp(\gamma') = \Sp(\gamma)+x\}.
	\ee
	Thus, we can partition the set $\mathcal{C}_0(m)$ into
	\[
	\mathcal{C}_0(m) = \bigcup_{\substack{0 \in \Sp(\gamma) \\ |\gamma|=m}}\mathcal{C}_\gamma.
	\]	
Given a contour $\gamma \in \mathcal{E}^-$, there is at most $|V(\gamma)|$ possibilities for the  position of the point $0$.
 Then,
	\be\label{prop5:eq1}
	|\mathcal{C}_0(m)| \leq \sum_{\substack{0 \in \Sp(\gamma), \gamma 
 \in \mathcal{E}^-\\ |\gamma|=m}}|\mathcal{C}_\gamma| \leq \sum_{\substack{0 \in \Sp(\gamma), \gamma \in \mathcal{E}^- \\ |\gamma|=m}}|V(\gamma)|.
	\ee
	Using the isoperimetric inequality and the fact $\din V(\gamma) \subset \Sp(\gamma)$, we obtain,
	\be\label{eq14}
	\sum_{\substack{0 \in \Sp(\gamma), \gamma \in \mathcal{E}^- \\ |\gamma|=m}}|V(\gamma)| \leq m^{1+ \frac{1}{d-1}}|\{\Sp(\gamma)\Subset \Z^d:\gamma \in \mathcal{E}^-, 0 \in \Sp(\gamma), |\gamma|=m \}|.
	\ee
	By Proposition \ref{prop4}, and since not all the finite sets with bounded total volume are contours, we have
	\be\label{eq15}
	\{\Sp(\gamma)\Subset \Z^d:\gamma \in \mathcal{E}^-, 0 \in \Sp(\gamma), |\gamma|=m \} \subset \{\Lambda \Subset \Z^d: 0 \in \Lambda, V_r(\Lambda)\leq \kappa m\}.
	\ee
	Proposition \ref{appB:prop2} yields
	\be\label{eq16}
	|\{\Lambda \Subset \Z^d: 0 \in \Lambda, V_r(\Lambda)\leq \kappa m \}|= \sum_{V=1}^{\lfloor \kappa m\rfloor}|\mathscr{F}_V| \leq \frac{e^{b(\kappa m+1)}}{e^b-1}.
	\ee
	Substituting Inequalities \eqref{eq14}, \eqref{eq15} and \eqref{eq16} into Inequality \eqref{prop5:eq1}, we conclude
	\be
	|\mathcal{C}_0(m)| \leq m^{1+ \frac{1}{d-1}}\frac{e^{2b\kappa m+1}}{e^{b}-1} \leq e^{c_1m},
	\ee
	for 
	$c_1 = 2b\kappa+1+ (d-1)^{-1}$.
\end{proof}

\section{Phase Transition}


In this section, we prove that the long-range Ising model with decaying field undergoes a phase transition at low temperature when $\min\{\alpha-d,1\}<\delta<d$. When the magnetic field decays with power $\delta \geq d$, the result is straightforward. In fact, for $\delta>d$, the magnetic field is summable and, by a general result of Georgii (see Example 7.32 and Theorem 7.33 in \cite{Geo}), there is an affine bijection between the Gibbs measures of the Ising model with $h=0$. Then, the phase transition is already known in this case. For $\delta=d$ the sum $\sum_{x\in \Lambda}h_x$ can be bounded by $\log |\Lambda|$. This implies that $\sum_{x \in \Lambda} h_x = o(|\Lambda|^\varepsilon)$ for any $\varepsilon>0$. Thus, if we prove the phase transition for $\delta < d$, it is easy to extend to this case.

\begin{theorem}\label{thm1}
	For a fixed $d\ge 2$, suppose that $\alpha>d$ and $\delta>0$.
	There exists $\beta_c\coloneqq \beta_c(\alpha,d)>0$ such that, for every $\beta>\beta_c$, the long-range Ising model with coupling constant  (\ref{long}) and magnetic field (\ref{magfield}) undergoes a phase transition at inverse of temperature $\beta$ when
	\begin{itemize}
		\item $d<\alpha<d+1$ and $\delta >\alpha -d$; 
		\item $d<\alpha<d+1$ and $\delta=\alpha-d$ if $h^*$ is small enough;
		\item $\alpha \geq d+1$ and $\delta>1$;
		\item  $\alpha \geq d+1$ and $\delta=1$ if $h^*$ is small enough.
	\end{itemize}.
\end{theorem}

For a fixed $x \in \Z^d$, define the function $\Theta_x : \Omega \rightarrow \mathbb{R}$ by
\[
\Theta_x(\sigma) = \prod_{\substack{y \in \Z^d \\ |x-y| \leq 1}}\mathbbm{1}_{\{\sigma_y = +1\}}- \prod_{\substack{y \in \Z^d \\ |x-y| \leq 1}}\mathbbm{1}_{\{\sigma_y = -1\}}. 
\]

The function $\Theta_x$ assumes the value $+1$ if the point $x$ is $+$-correct, $-1$ if the point is $-$-correct, and $0$ otherwise. By the definition of  contours, given a finite $\Lambda \Subset \Z^d$ and a configuration $\sigma \in \Omega^-_\Lambda$ it may happen that a contour $\gamma$ associated to it have volume outside $\Lambda$. To avoid this problem consider the probability measure defined as,
\be
\nu^-_{\beta,\mathbf{h},\Lambda}(A) \coloneqq \mu^-_{\beta,\mathbf{h},\Lambda}(A|\Theta_x =  - 1, x \in \din \Lambda),
\ee
for every measurable set $A$. The Markov property implies that the probability measures $\nu^-_{\beta,\mathbf{h},\Lambda}$ are the finite volume Gibbs measure in a subset of $\Lambda$ and to work with them is advantageous since we can study important quantities in terms of contours. Fixed $\Lambda \Subset \Z^d$, for each $\Lambda' \subset \Lambda$, let the restricted partition functions be
\be
Z_{\beta,\mathbf{h}}^-(\Lambda') \coloneqq \sum_{\substack{\Gamma \in \mathcal{E}_\Lambda^-\\ V(\Gamma)\subset \Lambda'}}\sum_{\sigma \in \Omega(\Gamma)}e^{-\beta H_{\Lambda, \mathbf{h}}^-(\sigma)},
\ee

where, given $\Gamma \in \mathcal{E}^-_\Lambda$, the space of configurations is $\Omega(\Gamma)\coloneqq \{\sigma \in \Omega^-_\Lambda: \Gamma \subset \Gamma(\sigma)\}$. Note that we will abuse the notation and denote $\Gamma(\sigma)$ as the set of contours associated to a configuration $\sigma$ instead of the $(M,a,r)$-partition. Define the map $\tau_\Gamma:\Omega(\Gamma) \rightarrow \Omega_{\Lambda}^-$ as
\be
\tau_\Gamma(\sigma)_x = 
\begin{cases}
	\;\;\;\sigma_x &\text{ if } x \in \I_-(\Gamma)\cup V(\Gamma)^c, \\
	-\sigma_x &\text{ if } x \in \I_+(\Gamma),\\
	-1 &\text{ if } x \in \Sp(\Gamma).
\end{cases}
\ee
The map $\tau_\Gamma$ erases a family of compatible contours, since the spin flip preserves incorrect points but transforms $+$-correct points into $-$-correct points. Given $\Gamma \in \mathcal{E}^-_\Lambda$ and a configuration $\sigma \in \Omega(\Gamma)$, we will calculate the energy cost to extract one of its elements. We can start with bounding the number of integer points in the $\ell^1$-sphere and continue with giving a lower bound for the diameter of a finite subset of $\Z^d$.
\begin{lemma}\label{comblema}
	Let $s_d(n)$ be the cardinality of integer points in the $\ell^1$-sphere, centered at the origin and with radius $n$. Then, for any $n\geq d$, we have,
	\[
	s_d(n) = \sum_{k=0}^{d-1}2^{d-k}\binom{d}{k}\binom{n-1}{d-k-1}.
	\]
	If $n<d$, the sum above starts in $k=d-n$. 
\end{lemma}
This Lemma allows us to show that 
\be\label{eq:desi_sd}
c_d n^{d-1}\leq s_d(n) \leq 2^{2d-1} e^{d-1} n^{d-1}, 
\ee
for all $n \geq d$ and $c_d = (d-1)^{-(d-1)}$.
\begin{lemma}\label{diamlema}
	There exists $k_d>0$ such that for every $\Lambda \Subset \Z^d$ it holds,
	\be\label{diamlema:eq1}
	\diam(\Lambda)\geq k_d |\Lambda|^{\frac{1}{d}}.
	\ee
\end{lemma}
\begin{proof}
	Consider a closed ball with positive integer radius $n$. Lemma \ref{comblema} implies that the diameter satisfies
	\[
	\diam(B_n(x))= 2n \geq C_d |B_n(x)|^{\frac{1}{d}},
	\]
	where $C_d = (2^{2d}e^{d-1}\frac{1}{d})^{-\frac{1}{d}}$. Let $\Lambda$ be any finite subset of $\Z^d$. If we take $x^*, y^* \in \Lambda$ such that $\diam(\Lambda)=|x^*-y^*|$, we have
	\[
	2\diam(\Lambda)= \diam( B_{|x^*-y^*|}(x^*)) \geq C_d|\Lambda|^{\frac{1}{d}}.
	\]
	Inequality \eqref{diamlema:eq1} follows by choosing the constant $k_d = C_d/2$. 
\end{proof}

In the next proposition, we will give a lower bound for the cost of extracting a contour from a given configuration. The main difference in our results compared to the short-range case is that we have an additional surface order term, defined as
\[
F_\Lambda \coloneqq \sum_{\substack{x \in \Lambda \\ y \in \Lambda^c}}J_{xy},
\]
for every finite set $\Lambda \Subset \Z^d$. First, let us give a lower bound to the surface energy term, that will be useful to the proof of phase transition.
\begin{lemma}\label{lema3}
	Given $\alpha>d$, there exists $K_\alpha\coloneqq K_\alpha(\alpha,d)>0$ such that for every $\Lambda \Subset \Z^d$ it holds
	\be
	F_\Lambda \geq K_{\alpha}\max\{|\Lambda|^{2-\frac{\alpha}{d}}, |\partial\Lambda|\}.
	\ee
\end{lemma}
\begin{proof}
	Since all the edges of $\partial\Lambda$ are present in the surface energy term $F_\Lambda$, we have the bound $F_\Lambda \geq J |\partial \Lambda|$. Fix $x \in \Lambda$.
	If we set $R = \lceil (dc_d^{-1}|\Lambda|)^{\frac{1}{d}} \rceil$ and using that $\sum_{\substack{y \in \Z^d}}J_{xy} < \infty$, we have
	\[
	\ssum{y \in \Lambda^c}J_{xy} - \ssum{y \in B_{R}(x)^c}J_{xy}=\ssum{y \in B_{R}(x)}J_{xy} - \ssum{y \in \Lambda}J_{xy} \geq \frac{J(|B_{R}(x)| - |\Lambda|)}{R^\alpha} \geq 0.
	\]
	 Lema \ref{comblema} yields us
	\[
	\sum_{y \in B_R(x)^c}J_{xy}  = J\sum_{n\geq R+1}\frac{s_d(n)}{n^\alpha} \geq J c_d\sum_{n \geq R+1} \frac{1}{n^{\alpha -d +1}}.
	\]
	The r.h.s can be bounded below by an integral, and we take $K_\alpha = Jc_d (\alpha-d)^{-1}((dc_d^{-1})^{1/d}+2)^{d-\alpha}$.
\end{proof}

In the next proposition we will denote $H_\Lambda^-$ the Hamiltonian function in \eqref{Isingsys} when the field $h_x = 0$, for all $x \in \Z^d$.
\begin{proposition}\label{prop2}
	For $M$ large enough, there are constants $c_i\coloneqq c_i(\alpha,d)>0$, $i=2,3,4$, such that for $\Lambda \Subset \Z^d$, any fixed contour $\gamma \in \mathcal{E}^-_\Lambda$, and $\sigma \in \Omega(\gamma)$ it holds 
	\be
	H_\Lambda^-(\sigma)- H_\Lambda^-(\tau(\sigma))\geq c_2|\gamma|+ c_3 F_{\I_+(\gamma)} + c_4 F_{\Sp(\gamma)}.
	\ee
\end{proposition}
\begin{proof}
	Fix some $\sigma \in \Omega(\gamma)$. We will denote $\tau_\gamma(\sigma)\coloneqq \tau$ and $\Gamma(\sigma) \coloneqq \Gamma$ throughout this proposition. The difference between the Hamiltonians is
	\begin{align}\label{prop2:eq1}
		H_{\Lambda}^-(\sigma)-H_{\Lambda}^-(\tau)&= \sum_{\{x,y\}\subset V(\Gamma)}J_{xy}(\tau_x\tau_y - \sigma_x\sigma_y) + \sum_{\substack{x \in V(\Gamma)\\ y \in V(\Gamma)^c}}J_{xy}(\sigma_x - \tau_x) \nonumber \\
		&=\sum_{\substack{x \in A(\gamma) \\ y \in B(\gamma)}}J_{xy}(\tau_x\tau_y - \sigma_x\sigma_y) + \sum_{\{x,y\} \subset A(\gamma)}J_{xy}(\tau_x\tau_y-\sigma_x\sigma_y),
	\end{align}
	where $A(\gamma)=\I_+(\gamma)\cup\Sp(\gamma)$ and $B(\gamma)=\I_-(\gamma)\cup V(\gamma)^c$. Using the definition of the map $\tau_\gamma$ and decomposing the sets $A(\gamma)$ and $B(\gamma)$ in a suitable manner, we get
	\begin{align}\label{prop2:eq2}
	H_{\Lambda}^-(\sigma)-H_{\Lambda}^-(\tau) &= \sum_{\substack{x \in \Sp(\gamma)\\ y \in \Z^d}}J_{xy}\mathbbm{1}_{\{\sigma_x \neq \sigma_y\}}+\sum_{\substack{x \in \Sp(\gamma)\\ y \in \Sp(\gamma)^c}}J_{xy}\mathbbm{1}_{\{\sigma_x \neq \sigma_y\}}-2\sum_{\substack{x \in I_+(\gamma) \\ y \in B(\gamma)}}J_{xy}\sigma_x\sigma_y \nonumber\\
	&- 2\sum_{\substack{x \in \Sp(\gamma) \\ y \in B(\gamma)}}J_{xy}\mathbbm{1}_{\{\sigma_y = +1\}}-2\sum_{\substack{x \in \Sp(\gamma) \\ y \in \I_+(\gamma)}}J_{xy}\mathbbm{1}_{\{\sigma_y = -1\}}.
	\end{align}
	We need to analyse each negative term of the equality above carefully. We start with the terms depending on $\Sp(\gamma)$. Notice that the characteristic functions on $B(\gamma)$ and $\I_+(\gamma)$ can only be different from zero at the other contours volumes. Thus,
	\be\label{prop2:eq3}
	\sum_{\substack{x \in \Sp(\gamma) \\ y \in B(\gamma)}}J_{xy}\mathbbm{1}_{\{\sigma_y = +1\}}+\sum_{\substack{x \in \Sp(\gamma) \\ y \in \I_+(\gamma)}}J_{xy}\mathbbm{1}_{\{\sigma_y = -1\}} \leq \sum_{\substack{x \in \Sp(\gamma) \\ y \in V(\Gamma')}}J_{xy},
	\ee
	where $\Gamma'\coloneqq \Gamma(\tau)$. Let $\gamma = \bigcup_{1 \leq k\leq n} \gamma_k$ and $\gamma'= \cup_{1\leq j\leq n'}\gamma'_j$ for each $\gamma' \in \Gamma'$ be the subsets given to us by condition \textbf{(B)}. We will divide the r.h.s of Equation \eqref{prop2:eq3} into two terms depending on the sets
	$\Upsilon_1= \{\gamma' \in \Gamma': \underset{1\leq k \leq n}{\max}\diam(\gamma_k)\leq \underset{1\leq j \leq n'}{\max}\diam(\gamma'_j)\}$  and $\Upsilon_2= \Gamma' \setminus \Upsilon_1$. On the first sum, Condition \textbf{(B2)} implies
	\[\label{prop2:eq4}
	\sum_{\substack{x \in \Sp(\gamma) \\ y \in V(\Upsilon_1)}}J_{xy} \leq \sum_{\substack{x \in \Sp(\gamma) \\ y \in B_R(x)^c}}J_{xy},
	\]
	where $R= M\underset{1\leq k \leq n}{\max}\diam(\gamma_k)^a$. Using Condition \textbf{(B1)} it holds,
	\be\label{prop2:eq5}
	\sum_{\substack{x \in \Sp(\gamma)\\ y \in B_R(x)^c}}J_{xy} \leq \frac{J2^{d+\alpha-1}e^{d-1}|\gamma|}{(\alpha-d)M^{\alpha-d}}\underset{1\leq k \leq n}{\max}\diam(\gamma_k)^{a(d-\alpha)}\leq  \frac{J2^{d+\alpha-1}e^{d-1}(2^r-1)}{(\alpha-d)k_d^{a(\alpha-d)}M^{\alpha -d}}.
	\ee
	We turn our attention to the sum depending on $\Upsilon_2$. We divide the set $\Upsilon_2$ into sets $\Upsilon_{2,m}$ consisting of contours of $\Gamma'$ where the maximum diameter of its partition is $m$. Thus, for any $x$ in $\Sp(\gamma)$ and $\gamma' \in \Gamma'$, there is a point $y_{\gamma',x}\in V(\gamma')$ such that $|x-y_{\gamma',x}| = \dis(\{x\},\gamma')$. Then, for every $1\leq m < \max_{1\leq k\leq n} \diam(\gamma_k)$,
	\[
	\sum_{\substack{x \in \Sp(\gamma) \\ y \in V(\Upsilon_{2,m})}}J_{xy} \leq \sum_{\substack{x \in \Sp(\gamma) \\ \gamma' \in \Upsilon_{2,m}}}|V(\gamma')|J_{xy_{\gamma',x}}.
	\]
	 For each $\gamma' \in \Upsilon_{2,m}$, define the graph $G_{\gamma'}$ with vertex set $v(G_{\gamma'})=\{\gamma'_j\}_{1\leq j \leq n'}$ and an edge is placed when $\dis(\gamma'_j,\gamma'_i)\leq 1$. Let $G_j$ be the maximal connected component of $G_{\gamma'}$ such that $\gamma'_j$ is an element of its vertex set. Also, let $V(G_j)\subset V(\gamma')$ be the subset of all connected components of $V(\gamma')$ that have a non empty intersection with the vertices of $G_j$. Using Lemma \ref{diamlema}, we have
	\[
	|V(\gamma')| \leq \frac{1}{k_d^d}\sum_{j=1}^{n'}\diam(V(G_j))^d.
	\]
	The diameter of $V(G_j)$ is realized by the distance between two points, namely $x^*$ and $y^*$, that must be into $\gamma'_i, \gamma'_l \in v(G_j)$. We can make a minimal path $\lambda_1$ in the graph $G_j$ between $\gamma'_i$ and $\gamma'_l$ since it is connected. Thus, using the path $\lambda_1$, we can construct a minimal path $\lambda_2$ in $\Z^d$ connecting $x^*$ and $y^*$ that passes through every vertex that is visited by the path $\lambda_1$. Let $\lambda_3$ be a minimal path realizing the distance between $x^*$ and $y^*$. Since the path is minimal, we have
	\[
	\diam(V(G_j))=|\lambda_3|\leq |\lambda_2|\leq \sum_{\gamma'_i \in v(G_j)}\left[\diam(\gamma'_i)+1\right].
	\]
	Hence,
	\be\label{eqvolume}
	\frac{1}{k_d^d}\sum_{j=1}^{n'}\diam(V(G_j))^d\leq \frac{1}{k_d^d}\sum_{j=1}^{n'} \left(\sum_{\gamma'_i \in v(G_j)}\left[\diam(\gamma'_i)+1\right]\right)^d.
	\ee
	The number of elements in $v(G_j)$ is at most $2^r-1$ by condition \textbf{(B)}, thus
	\be\label{prop2:eq6}
	\sum_{\substack{x \in \Sp(\gamma) \\ y \in V(\Upsilon_{2,m})}}J_{xy} \leq \frac{(2(2^r-1))^{d+1}}{k_d^d}m^d\sum_{\substack{x \in \Sp(\gamma) \\ \gamma' \in \Upsilon_{2,m}}}J_{xy_{\gamma',x}}.
	\ee
	We know that there is no other point $y_{\gamma'',x}$ at least in a ball of radius $Mm^a$ centered at $y_{\gamma',x}$. Thus, taking balls of radius $Mm^a/3$ guarantees that they become disjoint by Condition \textbf{(B2)}. Also, if $\lambda$ is the minimal path realizing the distance between $x$ and $y_{\gamma',x}$, we know that it must have at least $Mm^a$ points (see Figure \ref{figure_minimalpath}). Thus,
	\be\label{prop2:eq7}
	\sum_{\substack{x \in \Sp(\gamma) \\ \gamma' \in \Upsilon_{2,m}}}J_{xy_{\gamma',x}} \leq \frac{3}{Mm^a}F_{\Sp(\gamma)}.
	\ee
	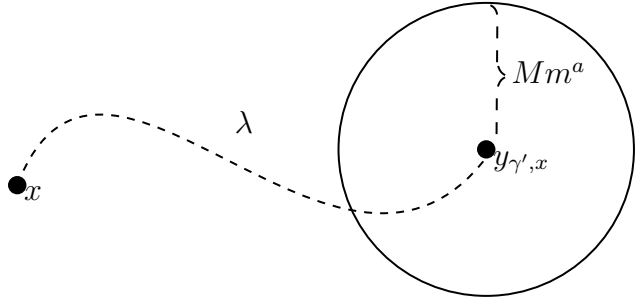
\begin{figure}[H]\label{figure_minimalpath}
		\centering
		
		\tikzset{every picture/.style={line width=0.75pt}} 
		
		\begin{tikzpicture}[x=0.75pt,y=0.75pt,yscale=-1,xscale=1]
			
			\draw  [fill={rgb, 255:red, 0; green, 0; blue, 0 }  ,fill opacity=1 ] (4.8,99.1) .. controls (4.8,96.84) and (6.64,95) .. (8.9,95) .. controls (11.16,95) and (13,96.84) .. (13,99.1) .. controls (13,101.36) and (11.16,103.2) .. (8.9,103.2) .. controls (6.64,103.2) and (4.8,101.36) .. (4.8,99.1) -- cycle ;
			\draw  [fill={rgb, 255:red, 0; green, 0; blue, 0 }  ,fill opacity=1 ] (238.8,81.1) .. controls (238.8,78.84) and (240.64,77) .. (242.9,77) .. controls (245.16,77) and (247,78.84) .. (247,81.1) .. controls (247,83.36) and (245.16,85.2) .. (242.9,85.2) .. controls (240.64,85.2) and (238.8,83.36) .. (238.8,81.1) -- cycle ;
			\draw [dashed  ,draw opacity=1 ]   (8.9,99.1) .. controls (51.4,-12) and (169.5,182) .. (242.9,85.2) ;
			\draw   (169.2,81.1) .. controls (169.2,40.4) and (202.2,7.4) .. (242.9,7.4) .. controls (283.6,7.4) and (316.6,40.4) .. (316.6,81.1) .. controls (316.6,121.8) and (283.6,154.8) .. (242.9,154.8) .. controls (202.2,154.8) and (169.2,121.8) .. (169.2,81.1) -- cycle ;
			\draw  [dash pattern={on 4.5pt off 4.5pt}] (241,81) .. controls (245.67,81.03) and (248.02,78.72) .. (248.05,74.05) -- (248.18,54.25) .. controls (248.23,47.58) and (250.58,44.27) .. (255.25,44.3) .. controls (250.58,44.27) and (248.27,40.92) .. (248.32,34.25)(248.3,37.25) -- (248.45,15.05) .. controls (248.48,10.38) and (246.17,8.03) .. (241.5,8) ;
			
			\draw (244.9,80.4) node [anchor=north west][inner sep=0.75pt]    {$y_{\gamma',x}$};
			\draw (10.9,98.4) node [anchor=north west][inner sep=0.75pt]    {${\displaystyle x}$};
			\draw (116,60.4) node [anchor=north west][inner sep=0.75pt]    {$\lambda $};
			\draw (254,35.4) node [anchor=north west][inner sep=0.75pt]    {$Mm^{a}$};
			
		\end{tikzpicture}
		\caption{Minimal path $\lambda$ between $x$ and $y_{\gamma',x}$.}
	\end{figure}
	Inequalities \eqref{prop2:eq7}, \eqref{prop2:eq6}, \eqref{prop2:eq5} plugged into Inequality \eqref{prop2:eq3} yields,
	\be\label{prop2:eq8}
	\sum_{\substack{x \in \Sp(\gamma) \\ y \in B(\gamma)}}J_{xy}\mathbbm{1}_{\{\sigma_y = +1\}}+\sum_{\substack{x \in \Sp(\gamma) \\ y \in \I_+(\gamma)}}J_{xy}\mathbbm{1}_{\{\sigma_y = -1\}} \leq \frac{k^{(1)}_\alpha}{M^{(\alpha-d)\wedge 1}} F_{\Sp(\gamma)}
	\ee
	where, $k^{(1)}_\alpha =2 \max\left\{\frac{J2^{d+\alpha-1}e^{d-1}(2^r-1)}{(\alpha-d)k_d^{a(\alpha-d)}},\frac{(2(2^r-1))^{d+1}3\zeta(a-d)}{k_d^d}\right\}$ and $(\alpha-d)\wedge 1 = \min\{\alpha-d,1\}$. 
	
	The remaining term in our analysis is the one involving the interaction between $\I_+(\gamma)$ and $B(\gamma)$. Recall that $\Gamma(\tau )=\Gamma'$ is the set of external contours of $\Gamma$ after $\gamma$ is removed and define $\Gamma_1\subset \Gamma'$ as the set of external contours that are contained in $\I_+(\gamma)$ and $\Gamma_2 = \Gamma' \setminus \Gamma_1$. We have,
	\begin{align}\label{prop2:eq9}
		\sum_{\substack{x \in \I_+(\gamma) \\ y \in B(\gamma)}}J_{xy}\sigma_x\sigma_y= \sum_{\substack{x \in V(\Gamma_1) \\ y \in V(\Gamma_2)}}J_{xy} +\sum_{\substack{x \in \I_+(\gamma)\setminus V(\Gamma_1) \\ y \in V(\Gamma_2)}}2J_{xy}\mathbbm{1}_{\{\sigma_y=+1\}} + \sum_{\substack{x \in V(\Gamma_1) \\ y \in B(\gamma)\setminus V(\Gamma_2)}}2J_{xy}\mathbbm{1}_{\{\sigma_x=-1\}} \nonumber \\
		-\sum_{\substack{x \in \I_+(\gamma)\setminus V(\Gamma_1) \\ y \in V(\Gamma_2)}}J_{xy}-\sum_{\substack{x \in V(\Gamma_1) \\ y \in V(\Gamma_2)}}2J_{xy}\mathbbm{1}_{\{\sigma_x \neq \sigma_y\}}-\sum_{\substack{x \in \I_+(\gamma) \\ y \in B(\gamma)\setminus V(\Gamma_2)}}J_{xy}.
	\end{align}
	We start our analysis with the first two terms on r.h.s of \eqref{prop2:eq9}. Note that,
	\be\label{again}
	\sum_{\substack{x \in V(\Gamma_1) \\ y \in V(\Gamma_2)}}J_{xy} + \sum_{\substack{x \in \I_+(\gamma)\setminus V(\Gamma_1) \\ y \in V(\Gamma_2)}}2J_{xy}\mathbbm{1}_{\{\sigma_y=+1\}}\leq 2 \sum_{\substack{x \in \I_+(\gamma) \\ y \in V(\Gamma_2)}}J_{xy}.
	\ee
	Consider the two sets $\Upsilon_3=\{\gamma' \in \Gamma_2: \underset{1\leq k \leq n}{\max}\diam(\gamma_k) \leq \underset{1\leq j \leq n'}{\max}\diam(\gamma'_j)\}$ and $\Upsilon_4 = \Gamma_2 \setminus \Upsilon_3$. Now we proceed as in the previous case for $\Upsilon_1$,
	\be\label{eq11111}
	\sum_{\substack{x \in \I_+(\gamma)\\y \in V(\Upsilon_3)}}J_{xy} \leq \frac{J2^{d+\alpha-1} e^{d-1}}{(\alpha-d)M^{\alpha-d}} |\I_+(\gamma)|\underset{1\leq k \leq n}{\max}\diam(\gamma_k)^{a(d-\alpha)}.
	\ee
    Notice that by using the isoperimetric inequality together with Lemma \ref{diamlema}, we get that 
    \[
    |\I_+(\gamma)|\underset{1\leq k \leq n}{\max}\diam(\gamma_k)^{a(d-\alpha)} \leq \frac{|\din \gamma|^{\frac{d}{d-1}}}{k_d^{a(\alpha-d)}} \underset{1\leq k \leq n}{\max}|\gamma_k|^{\frac{a(d-\alpha)}{d}}\leq \frac{(2^r-1)}{k_d^{a(\alpha-d)}} |\din \gamma|^{\frac{1}{d-1}}\underset{1\leq k \leq n}{\max}|\gamma_k|^{\frac{a(d-\alpha)}{d}+1}.
    \]
    Plugging the above inequality into \eqref{eq11111} and using our choice of the constant $a$ we get
    \be\label{prop2:eq10}
	\sum_{\substack{x \in \I_+(\gamma)\\y \in V(\Upsilon_3)}}J_{xy} \leq \frac{k_\alpha^{(1)}}{M^{\alpha-d}}|\din\gamma|^{\frac{1}{d-1}}\leq \frac{k_\alpha^{(1)}}{M^{\alpha-d}}|\gamma|.
	\ee
    For the sum depending on the contours in $\Upsilon_4$, we will need to break, as before, into sets $\Upsilon_{4,m}$ whose contours have maximum diameter equals to $m$. An argument similar to the one employed in \eqref{prop2:eq8} holds, hence
	\be\label{prop2:eq11}
	\sum_{\substack{x \in \I_+(\gamma)\\y \in V(\Upsilon_4)}}J_{xy} \leq \frac{(2(2^r-1))^{d+1}3\zeta(a-d)}{k_d^d M}\sum_{\substack{x \in \I_+(\gamma) \\ y \in \I_-(\gamma)\cup V(\gamma)^c}}J_{xy}\leq \frac{k_\alpha^{(1)}}{M}F_{\I_+(\gamma)}.
	\ee
	For the next term, since we have $B(\gamma)\setminus V(\Gamma_2) \subset \I_+(\gamma)^c$, we get
	\[
	\sum_{\substack{x \in V(\Gamma_1) \\ y \in B(\gamma)\setminus V(\Gamma_2)}}J_{xy} \leq \sum_{\substack{x \in V(\Gamma_1) \\ y \in \I_+(\gamma)^c}}J_{xy}.
	\] 
	We claim that, for any $\gamma' \in \Gamma_1$, $\max_{1\leq j \leq n'}\diam(\gamma'_j) < \max_{1\leq k \leq n}\diam(\gamma_k)$ for $M>(2(2^r-1))^{d+1}/k_d^d$. Indeed, by condition \textbf{(A)}, $\Sp(\gamma')$ is contained in only one connected component of $\I(\gamma)$, let us call it $\I_+(\gamma)^{(1)}$. By similar reasonings as the one that gave us Inequality \eqref{eqvolume}, we have
	  \[
	  |\I_+(\gamma)^{(1)}|\leq \frac{(2(2^r-1))^{d+1}}{k_d^d}\max_{1\leq k \leq n}\diam(\gamma_k)^d.
	  \] 
	  Assume by contradiction $\max_{1\leq j \leq n'}\diam(\gamma'_j) \geq \max_{1\leq k \leq n}\diam(\gamma_k)$, then Condition \textbf{(B2)} implies that $\dis(\gamma',\gamma)\geq M \max_{1 \leq k \leq n}\diam(\gamma_k)^a$. Therefore, $|\I_+(\gamma)|$ must have at least $M \max_{1 \leq k \leq n}\diam(\gamma_k)^a$ points inside it, which is a contradiction with our choice of $M$.
	
	Thus, let us break $\Gamma_1$ into layers $\Gamma_{1,m}$ where $\max_{1\leq j \leq n'}\diam(\gamma'_j)=m$. For each $y \in \I_+(\gamma)^c$ and $\gamma' \in \Gamma_1$ there is $x_{\gamma',y} \in V(\gamma')$ that realizes the distance between $V(\gamma')$ and $y$. Hence,
	\be\label{eqfinal}
	\sum_{\substack{x \in V(\Gamma_1) \\ y \in B(\gamma)\setminus V(\Gamma_2)}}J_{xy}\leq \sum_{m = 1}^{N-1}\sum_{\substack{x \in V(\Gamma_{1,m}) \\ y \in \I_+(\gamma)^c}}J_{xy} \leq \sum_{m = 1}^{N-1}\frac{(2(2^r-1))^{d+1}}{k_d^d}m^d \sum_{\substack{\gamma' \in \Gamma_{1,m} \\ y \in \I_+(\gamma)^c}}J_{x_{\gamma',y}y} \leq \frac{k_\alpha^{(1)}}{M} F_{\I_+(\gamma)},
	\ee
	where $N\coloneqq \max_{1\leq k\leq n}\diam(\gamma_k)$. We turn our attention to the term containing $\mathbbm{1}_{\{\sigma_x \neq \sigma_y\}}$ in the r.h.s of Inequality \eqref{prop2:eq9}. The triangle inequality implies that the following inequality holds
	\be\label{prop2:eq12}
	J_{xy} \geq \frac{1}{(2d+1)2^\alpha}\sum_{|x-x'| \leq 1}J_{x'y},
	\ee
	for every distinct pair of points $x,y \in \Z^d$. Thus, we have that
	\be\label{prop2:eq13}
	\sum_{\substack{x \in V(\Gamma_1)\\ y \in V(\Gamma_2)}}J_{xy}\mathbbm{1}_{\{\sigma_x \neq \sigma_y\}}\geq \frac{1}{(2d+1)2^\alpha}\sum_{\substack{x \in V(\Gamma_1)_0 \\ y \in V(\Gamma_2)}}J_{xy},
	\ee
	where $V(\Gamma_1)_0 = \{x \in V(\Gamma_1): \Theta_x(\sigma)=0\}$. Plugging Inequalities \eqref{prop2:eq10}, \eqref{prop2:eq11}, \eqref{eqfinal}, \eqref{prop2:eq13} into Equation \eqref{prop2:eq9}, we get
	\begin{align}\label{prop2:eq14}
	\sum_{\substack{x \in \I_+(\gamma)\\ y \in B(\gamma)}}J_{xy}\sigma_x\sigma_y &\leq \frac{4k^{(1)}_\alpha}{M}F_{\I_+(\gamma)} +\frac{2(2d)^{\frac{1}{d-1}}k_\alpha^{(1)}}{M^{\alpha-d}}|\gamma|-\frac{1}{(2d+1)2^{\alpha-1}}\sum_{\substack{x \in V(\Gamma_1)_0 \\ y \in V(\Gamma_2)}}J_{xy} \nonumber \\
	&- \sum_{\substack{x \in \I_+(\gamma)\\ y \in B(\gamma)\setminus V(\Gamma_2)}}J_{xy}- \sum_{\substack{x \in \I_+(\gamma)\setminus V(\Gamma_1) \\ y \in  V(\Gamma_2)}}J_{xy}.
	\end{align}
	 We must add the regions with correct points into the sum depending on $V(\Gamma_1)_0$. But this is a simple task, since we have,
	\be
	\sum_{\substack{x \in V(\Gamma_1)\setminus V(\Gamma_1)_0\\ y \in V(\Gamma_2)}}J_{xy} \leq \sum_{\substack{x \in V(\Gamma_1)\\ y \in V(\Gamma_2)}}J_{xy},
	\ee
	and proceeding as we did in \eqref{again}, we arrive at the following inequality
	\be\label{prop2:eq15}
	\sum_{\substack{x \in \I_+(\gamma) \\ y \in B(\gamma)}}J_{xy}\sigma_x\sigma_y \leq \frac{5k^{(1)}_\alpha}{M}F_{\I_+(\gamma)} +\frac{3(2d)^{\frac{1}{d-1}}k_\alpha^{(1)}}{M^{\alpha-d}}|\gamma|-\frac{1}{(2d+1)2^{\alpha-1}}\sum_{\substack{x\in \I_+(\gamma) \\ y \in B(\gamma)}}J_{xy}.
	\ee
	Also, Inequality \eqref{prop2:eq12} implies that
	\begin{equation}\label{prop2:eq16}
        \begin{split}
	\sum_{\substack{x \in \Sp(\gamma)\\ y \in \Z^d}}J_{xy}\mathbbm{1}_{\{\sigma_x \neq \sigma_y\}}+&\sum_{\substack{x \in \Sp(\gamma)\\ y \in \Sp(\gamma)^c}}J_{xy}\mathbbm{1}_{\{\sigma_x \neq \sigma_y\}} \geq \frac{1}{(2d+1)2^\alpha}\left(Jc_\alpha|\gamma|+ F_{\Sp(\gamma)}\right) \\
        &\geq\frac{J c_\alpha}{(2d+1)2^\alpha}|\gamma|+ \frac{1}{(2d+1)2^{\alpha+1}}\bigg(F_{\Sp(\gamma)}+\sum_{\substack{x\in \Sp(\gamma) \\ y \in \I_+(\gamma)}}J_{xy}\bigg)
        \end{split}
	\end{equation}
	where $c_\alpha = \sum_{y \neq 0 \in \Z^d}\frac{1}{|y|^\alpha}$. Joining Inequalities \eqref{prop2:eq8},\eqref{prop2:eq15} and \eqref{prop2:eq16} into \eqref{prop2:eq2} yields
	\begin{align}
	H_\Lambda^-(\sigma)-H_\Lambda^-(\tau) \geq & \left( \frac{Jc_\alpha}{(2d+1)2^\alpha}- \frac{6(2d)^{\frac{1}{d-1}}k_\alpha^{(1)}}{M^{\alpha-d}}\right)|\gamma| + \left(\frac{1}{(2d+1)2^{\alpha+1}} -\frac{10k^{(1)}_\alpha}{M}\right)F_{\I_+(\gamma)} \nonumber \\
	&+\left(\frac{1}{(2d+1)2^{\alpha+1}}-\frac{2k_\alpha^{(1)}}{M^{(\alpha-d)\wedge1}}\right)F_{\Sp(\gamma)}.
	\end{align}
	Letting $M> \max\{\frac{(2(2^r-1))^{d+1}}{k_d^d}, M_1, M_2\}$, where $Jc_\alpha M_1^{\alpha-d}\coloneqq 36d(2d)^{\frac{1}{d-1}}k_\alpha^{(1)}2^{\alpha+2}$, ${M_2^{(\alpha-d)\wedge 1} \coloneqq 30dk_\alpha^{(1)}2^{\alpha+2}}$ we arrive at the desired result.
\end{proof}

The following Lemma is necessary to study the competition between the magnetic field and the long-range interaction.

\begin{lemma}\label{lema1}
	There exists $c_5\coloneqq c_5(d,\delta,h^*)>0$ such that for any $\Lambda \Subset \Z^d$ 
	\begin{equation}\label{eq4}
		\sum_{x \in \Lambda} h_x \leq c_5 |\Lambda|^{1-\frac{\delta}{d}},
	\end{equation}
 where $(h_x)_{x \in\Z^d}$ is the magnetic field as in (\ref{magfield}) and $\delta<d$. 
\end{lemma}
\begin{proof}
	Fix $\Lambda \Subset \Z^d$. In order to prove the inequality (\ref{eq4}), we will show that the sum in the l.h.s is always upper bounded by the sum of the magnetic field $h_x$ in some ball $B_R(0)$ with $R$ large enough. In fact, the magnetic field satisfies $h_x\geq h^*/R^\delta$ for $x \in B_R(0)$, and $h_x< h^*/R^\delta$ for $x \in \Lambda\setminus B_R(0)$. Then, we have
	\be\label{eq:diff_field}
	\sum_{x\in B_R(0)}h_x -\sum_{x \in \Lambda} h_x \geq  \frac{h^*}{R^\delta}\left( |B_R(0)|-|\Lambda|\right).
	\ee
	By using Lemma \ref{comblema}, we find that the volume of ball satisfies $|B_R(0)| \geq d^{-1}c_d R^d$. Thus, if we choose $R \coloneqq \lceil (dc_d^{-1}|\Lambda| )^{\frac{1}{d}}\rceil$, the r.h.s of Inequality \eqref{eq:diff_field} is nonnegative. We can bound the sum of the magnetic field in a ball $B_R(0)$ in the following way,
	\[
	\ssum{x \in B_R(0)} h_x \leq h^*+ h^*2^{2d-1}e^{d-1}\sum_{n=1}^{R} n^{d-1-\delta}.
	\]
	 The result is obtained taking $c_5 = h^*2^{2d} e^{d-1}(d-\delta)^{-1}((dc_d^{-1})^{1/d}+2)^{d-\delta}$.
\end{proof}

 As in the usual Peierls argument, Theorem \ref{thm1} will follow once we proof the following proposition.
\begin{proposition}\label{prop3}
	Let $\alpha>d$ and $\delta>0$. For $\beta$ large enough, it holds that 
	\be
	\nu^-_{\beta,\mathbf{h},\Lambda}(\sigma_0= + 1) < \frac{1}{2},
	\ee
	for every $\Lambda \Subset \Z^d$ when
	\begin{itemize}
		\item $d<\alpha<d+1$ and $\delta>\alpha -d$;
		\item $d<\alpha<d+1$ and $\delta=\alpha-d$ if $h^*$ is small enough;
		\item $\alpha \geq d+1$ and $\delta>1$;
		\item $\alpha \geq d+1$ and $\delta=1$ if $h^*$ is small enough.
	\end{itemize}
\end{proposition}
\begin{proof}
	Let $R>0$ and $(\widehat{h}_x)_{x\in\Z^d}$ be the truncated magnetic field 
	\begin{equation}\label{magfield2}
		\widehat{h}_x=\begin{cases}
			0 &\text{ if } |x|<R; \\
			h_x &\text{ if } |x| \geq R.
		\end{cases}
	\end{equation}
	The constant $R$ will be chosen later. Existence of phase transition under the presence of the truncated field implies phase transition for the model with the decaying field (see Theorem 7.33 of \cite{Geo} for a more general statement). If $\sigma_0=+1$ there must exist a contour $\gamma$ such that $0 \in V(\gamma)$. Hence
	\[\label{prop3:eq1}
	\nu^-_{\beta,\mathbf{\widehat{h}},\Lambda}(\sigma_0= + 1) \leq \sum_{\substack{\gamma \in \mathcal{E}_\Lambda^- \\ 0 \in V(\gamma)}} \nu^-_{\beta,\mathbf{\widehat{h}},\Lambda}(\Omega(\gamma)). 
	\]
	Using Proposition \ref{prop2}, we know that the Hamiltonian $H^-_{\Lambda,\mathbf{\widehat{h}}}$ satisfies,
	\be\label{prop3:eq2}
	H^-_{\Lambda,\mathbf{\widehat{h}}}(\sigma)-H^-_{\Lambda,\mathbf{\widehat{h}}}(\tau(\sigma)) \geq c_2 |\gamma| + c_3 F_{\I_+(\gamma)} - 2\sum_{x \in \I_+(\gamma)\cup \Sp(\gamma)}\hat{h}_x,
	\ee
	 where $\tau_\gamma(\sigma)= \tau(\sigma)$. Notice that 
	\begin{equation}\label{field_contribution}
	\sum_{x \in \Sp(\gamma)}\hat{h}_x \leq \frac{h^* |\gamma|}{R^\delta}.
	\end{equation}
	The choice $R^\delta > \frac{4h^*}{c_2}$ is sufficient to guarantee that the term $c_2 |\gamma|$ is larger than the field contribution term \eqref{field_contribution} in the right-hand side of Inequality \eqref{prop3:eq2}. The next step is to analyse when the term $c_3 F_{\I_+(\gamma)} - 2\sum_{x \in \I_+(\gamma)}\hat{h}_x$ is nonnegative. If $\I_+(\gamma) = \emptyset$, there is nothing to do, since the bound is trivial. Otherwise, we must analyse the competition of the decaying field with the different regimes of decay for the couplings constants $J_{xy}$.
	\begin{enumerate}[label=\textbf{(\roman*)}]
		\item \emph{Case} $d< \alpha < d+1$.
		By Lemmas \ref{lema3} and \ref{lema1}, we have 
		\be\label{prop3:eq3}
		c_3 F_{\I_+(\gamma)}- 2\sum_{x \in \I_+(\gamma)}\hat{h}_x \geq c_3K_\alpha |\I_+(\gamma)|^{2 - \frac{\alpha}{d}}- 2c_5|\I_+(\gamma)|^{1-\frac{\delta}{d}}.
		\ee
		Thus, if $\delta > \alpha - d$ and $|\I_+(\gamma)| \geq c'_\alpha \coloneqq \left(\frac{2c_5}{c_3 K_\alpha}\right)^{\frac{d}{\delta -(\alpha-d)}}$, we have that the r.h.s of Inequality \eqref{prop3:eq3} is nonnegative. In order to get a positive difference for all sizes of $\I_+(\gamma)$, we need to consider $R^\delta>R_1^\delta \coloneqq\frac{2h^* c'_\alpha}{c_3 K_\alpha}$. For the case $\delta=\alpha-d$, we must take $h^*$ small enough since the exponents in \eqref{prop3:eq3} will be equal.
		\item \emph{Case} $\alpha \geq d+1$.
		By Lemmas \ref{lema3} and \ref{lema1}, we have 
		\be\label{prop3:eq4}
		c_3 F_{\I_+(\gamma)}- 2\sum_{x \in \I_+(\gamma)}\hat{h}_x \geq c_3K_\alpha |\partial\I_+(\gamma)|- 2c_5|\I_+(\gamma)|^{1-\frac{\delta}{d}}.
		\ee
		Thus, if $\delta > 1$ and $|\I_+(\gamma)| \geq b_\alpha \coloneqq \left(\frac{c_5}{dc_3 K_\alpha}\right)^{\frac{d}{\delta-1}}$, we have that the r.h.s of Inequality \eqref{prop3:eq4} is nonnegative. In order to get a positive difference for all sizes of $\I_+(\gamma)$, we need to consider $R^\delta>R_2^\delta \coloneqq \frac{h^* b_\alpha}{dc_3 K_\alpha}$. The case where $\delta=1$, we must take $h^*$ small enough and use the isoperimetric inequality in Inequality \eqref{prop3:eq4}.
	\end{enumerate}
	It is clear that, by taking $R = \max\{\left(\frac{4h^*}{c_2}\right)^\frac{1}{\delta}, R_1, R_2\}$ together with \eqref{prop3:eq2} we get
	\[
	H^-_{\Lambda,\mathbf{\widehat{h}}}(\sigma)-H^-_{\Lambda,\mathbf{\widehat{h}}}(\tau(\sigma)) \geq \frac{c_2}{2} |\gamma|,
	\]
	which implies
	\be\label{eq112}
	\nu_{\beta,\mathbf{\widehat{h}},\Lambda}^-(\Omega(\gamma)) \leq \frac{e^{-\beta \frac{c_2}{2} |\gamma|}}{Z_{\beta, \mathbf{\widehat{h}}}^-(\Lambda)}\sum_{\sigma \in \Omega(\gamma)}e^{-\beta H^-_{\Lambda,\mathbf{\widehat{h}}}(\tau(\sigma))}.
	\ee
	Using the decomposition
	\[
	\Omega(\gamma) = \bigcup_{\Gamma: \Gamma \cup \{\gamma\} \in \mathcal{E}_\Lambda^-} \{ \sigma \in \Omega_\Lambda^-: \Gamma(\sigma)= \Gamma \cup \{\gamma\}\},
	\]
	together with the fact that, when we erase the contour $\gamma$, we may create new external contours but it always holds that $V(\Gamma(\tau(\sigma))) \subset \Lambda \setminus \Sp(\gamma)$. Hence, the r.h.s of Inequality \eqref{eq112} can be bounded as follows
	\begin{align*}
	\sum_{\sigma \in \Omega(\gamma)}e^{-\beta H^-_{\Lambda,\mathbf{\widehat{h}}}(\tau(\sigma))} &\leq \sum_{\substack{\Gamma \in \mathcal{E}^-_\Lambda \\ V(\Gamma)\subset \Lambda \setminus V(\gamma)}} \sum_{\substack{\Gamma' \in \mathcal{E}^-_\Lambda \\ V(\Gamma')\subset \I(\gamma)}}\sum_{\substack{\tau(\sigma) \\ \Gamma(\tau(\sigma))=\Gamma \cup \Gamma'}} \sum_{\omega: \tau(\omega)=\tau(\sigma)} e^{-\beta H^-_{\Lambda,\mathbf{\widehat{h}}}(\tau(\sigma))}\\
	&\leq |\{\sigma \in \Omega_{\Sp(\gamma)}: \Theta_x(\sigma)=0, \text{ for each } x \in \Sp(\gamma)\}| \sum_{\substack{\Gamma \in \mathcal{E}_{\Lambda}^- \\ V(\Gamma)\subset \Lambda \setminus \Sp(\gamma)}} \sum_{\sigma \in \Omega(\Gamma)}e^{-\beta H^-_{\Lambda,\mathbf{\widehat{h}}}(\sigma)}.
	\end{align*}
	Since the number of configurations that are incorrect in $\Sp(\gamma)$ are bounded by $2^{|\gamma|}$, we get
	\be\label{prop3:eq5}
	\nu_{\beta,\mathbf{\widehat{h}},\Lambda}^-(\Omega(\gamma)) \leq \frac{Z^-_{\beta, \mathbf{\widehat{h}}}(\Lambda\setminus\Sp(\gamma))e^{(\log(2)-\beta \frac{c_2}{2})|\gamma|}}{Z^-_{\beta, \mathbf{\widehat{h}}}(\Lambda)}.
	\ee
	Summing over all contours yields, together with Proposition \ref{importprop},
	\begin{align}
	\nu_{\beta,\mathbf{\widehat{h}},\Lambda}^-(\sigma_0 = +1) &\leq \sum_{\substack{\gamma \in \mathcal{E}_\Lambda^- \\ 0 \in V(\gamma)}}\frac{Z^-_{\beta, \mathbf{\widehat{h}}}(\Lambda\setminus\Sp(\gamma))e^{(\log(2)-\beta \frac{c_2}{2})|\gamma|}}{Z^-_{\beta, \mathbf{\widehat{h}}}(\Lambda)} \nonumber \\
	&\leq \sum_{m \geq 1}|\mathcal{C}_0(m)|e^{(\log(2)-\beta \frac{c_2}{2})m} \nonumber\\
	&\leq \sum_{m \geq 1}e^{(c_1 +\log(2) - \beta \frac{c_2}{2})m}< \frac{1}{2},
	\end{align}
	for $\beta$ large enough. 
\end{proof}

\section{Concluding Remarks}

In this paper, we developed a contour argument for phase transition in long-range Ising models when $d\geq 2$. As an application, we showed that the ferromagnetic Ising model with a decaying field presents a phase transition. Since for the borderline case $\delta=\alpha-d$ we need to consider $h^*$ small, this is an indication that phase transition should not hold further into the region of the exponents. Some piece of evidence of this phenomenon is given by the nearest-neighbor case studied in \cite{Bis2}, where uniqueness happened whenever $\delta<1$. 

Another natural question is to investigate if we can extend the argument to more general interactions in a sharp region of the exponents. One of such models is the ferromagnet nearest-neighbor Ising model with a competing long-range antiferromagnet interaction, as considered in \cite{Bisk}. As stated in their paper, a zero magnetization does not imply the absence of phase transition, and maybe some of the techniques developed here could be helpful to investigate the problem.

Moreover, our contours make sense for more general state spaces, so one could try to study other systems such as the Potts model, extending the results of \cite{Park1} and \cite{Park2}. We did that and also studied the cluster expansion and decay of correlations as was done in \cite{Imbr1, Imbr2} for the one-dimensional case in a separate paper which will appear soon.

\section*{Acknowledgements}

LA, RB, and EE thank Bruno Kimura for providing us with his thesis \cite{Kim} and many helpful discussions concerning his work in 1d long-range Ising models; to Aernout van Enter for many references and help with the literature. We thank Christian T\'{a}fula for suggesting a simpler proof for Proposition $3.12$. A sincere thank you to Jo\~{a}o Maia for his diligent proofreading of this paper. In addition, we thank the referee for the careful reading, which improved the paper. LA is supported by FAPESP Grant 2017/18152-2 and 2020/14563-0. RB is supported by CNPq Grant 312294/2018-2 and FAPESP Grant 16/25053-8; he thanks to Marcelo Disconzi for his support and his friendship over the years. SH is supported by the Ministry of Education, Culture, Sports, Science and Technology through the Program for Leading Graduate Schools (Hokkaido University Ambitious Leader's Program).

\end{document}